\documentclass[11pt]{article}

\usepackage[utf8]{inputenc}

\usepackage{amsmath,amsfonts,amsthm,amssymb,color}
\usepackage[usenames,dvipsnames,svgnames,table]{xcolor}
\definecolor{darkgreen}{rgb}{0.0,0,0.9}
\usepackage{mathtools}
\usepackage{authblk}
\usepackage{fullpage}
\usepackage{parskip}
\usepackage{comment}
\usepackage{tikz}
\usepackage{bbm}
\usepackage{dsfont}
\usepackage[sc]{mathpazo}
\usepackage[basic]{complexity}
\usepackage{algorithm2e}
\usepackage[colorlinks=true,
citecolor=OliveGreen,linkcolor=BrickRed,urlcolor=BrickRed,pdfstartview=FitH]{hyperref}
\usepackage[capitalize,nameinlink]{cleveref}
\usepackage{tcolorbox}
\usepackage[short]{optidef}

\newtcolorbox{wbox}
{
	colback  = white,
}

\usepackage{extras-alg22}

\newcommand{\e}{\mbox{\rm evenlevel}}
\renewcommand{\o}{\mbox{\rm oddlevel}}
\newcommand{\mn}{\mbox{\rm minlevel}}
\newcommand{\mx}{\mbox{\rm maxlevel}}
\renewcommand{\t}{\mbox{\rm tenacity}}
\newcommand{\bt}{\mbox{$\CB_{b,t}$}}
\newcommand{\base}{\mbox{\rm base}}
\newcommand{\CB}{\mbox{${\mathcal B}$}}

\newcommand{\bd}{\mbox{\rm bud}}
\newcommand{\bds}{\mbox{${\rm bud^*}$}}

\newcommand{\Ct}{\mbox{\bf Claim({\mbox{\boldmath $t$}})}}

\newcommand{\suppress}[1]{}

\makeatletter
\def\fnum@figure{{\bf Figure \thefigure}}
\def\fnum@table{{\bf Table \thetable}}

\long\def\@mycaption#1[#2]#3{\addcontentsline{\csname
 ext@#1\endcsname}{#1}{\protect\numberline{\csname
  the#1\endcsname}{\ignorespaces #2}}\par
     \begingroup
       \@parboxrestore
          \small
       \@makecaption{\csname fnum@#1\endcsname}{\ignorespaces
#3\endgroup}
      }

\newcommand{\A}{\alpha}
\newcommand{\B}{\beta}

\newcommand*{\Zplus}{\mathbb{Z_+}}

\renewcommand*{\cP}{{\mathcal{P}}}

\newcommand*{\la}{\leftarrow}
\newcommand{\support}{\mathrm{supp}}

\makeatletter
\def\thm@space@setup{%
	\thm@preskip= 10pt
	\thm@postskip=\thm@preskip 
}
\makeatother

\makeatletter
\renewcommand{\paragraph}{%
	\@startsection{paragraph}{4}%
	{\z@}{5pt}{-1em}%
	{\normalfont\normalsize\bfseries}%
}
\makeatother

\newenvironment{fminipage}%
{\begin{Sbox}\begin{minipage}}%
		{\end{minipage}\end{Sbox}\fbox{\TheSbox}}

\title{A Proof of the MV Matching Algorithm}

\author[1]{Vijay V.~Vazirani}

\affil[1]{University of California, Irvine}

\date{}

\begin{document}
	\maketitle

\begin{abstract}
The Micali-Vazirani (MV) algorithm for maximum cardinality matching in general graphs, which was published in 1980 \cite{MV}, remains to this day the most efficient known algorithm for the problem.

This paper gives the first complete and correct proof of this algorithm. Central to our proof are some purely graph-theoretic facts, capturing properties of minimum length alternating paths; these may be of independent interest. An attempt is made to render the algorithm easier to comprehend. 
\end{abstract}
    
\bigskip
\bigskip
\bigskip
\bigskip
\bigskip
\bigskip
\bigskip
\bigskip

\bigskip
\bigskip
\bigskip
\bigskip
\bigskip
\bigskip
\bigskip
\bigskip

\bigskip
\bigskip
\bigskip
\bigskip
\bigskip
\bigskip

\pagebreak
    
\section{Introduction}
\label{sec:intro}

The Micali-Vazirani (MV) \cite{MV} general graph maximum cardinality matching algorithm was published in 1980. It remains to this day the most efficient known algorithm for the problem; see Section \ref{sec.running} for precise details, including a slight speed-up for a special case, though of academic interest only. The description of the algorithm, given via a pseudo-code in \cite{MV}, is complete and error-free; however, the paper did not attempt a proof of correctness. 

The first polynomial time algorithm for this problem was given by Edmonds \cite{Edmonds}, central to which was the notion of {\em blossom}. Edmonds established purely graph-theoretic facts formalizing the manner in which alternating paths traverse blossoms and their complex nested structure. These facts formed the core around which his proof of correctness was built.

In the same vein, a complete proof of correctness of the MV algorithm requires graph-theoretic structural facts; see Section \ref{sec.contributions} for a precise reason for this substantial development. The process of formalizing these facts was initiated in \cite{va.matching}. As detailed in Section \ref{sec.contributions}, this paper made some important contributions; however, the structure is too rich and complex, and the attempt in \cite{va.matching} had serious shortcomings.

The current paper gives the first complete and correct proof of the MV algorithm, including purely graph-theoretic facts, capturing properties of minimum length alternating paths, which may be of independent interest. Additionally, innovative expository techniques are utilized to render the algorithm easier to comprehend. Considering the special status of this problem in the theory of algorithms, see below, it was not appropriate to leave this algorithm in an essentially unproven state --- hence the investment of (substantial) effort in the current paper.


Similar to the bipartite graph matching algorithms of \cite{Karp, Karzanov}, the MV algorithm also works by finding minimum length augmenting paths. However, a stark difference arises here: in non-bipartite graphs minimum length alternating paths do not possess an elementary property, called {\em breadth first search honesty}. Consequently, finding a minimum length augmenting path entails finding arbitrarily long paths to intermediate vertices, even though they admit short paths, see  Section \ref{sec.basic-definitions}. Since finding long paths, such as a Hamiltonian path, is NP-hard, the proposed approach appears to be hopeless. A recourse is provided by the remarkable structural properties of matching.

To give some feel for this structure, we will very informally describe next. We define the new notions of {\em base of a vertex} and the {\em tenacity of vertices and edges}. Building on these, we define the notion of a {\em blossom from the viewpoint of minimum length alternating paths}. The blossom of tenacity $t$ and base $b$ is denoted by $\bt$. The role of a blossom is to ``hide complexity within'': for each vertex $v \in \bt$, all minimum length alternating paths of both parity, even and odd, from $b$ to $v$ are guaranteed to lie entirely inside the blossom. Furthermore, the interface of a blossom with the ``outside world'' is very simple, see Theorem  \ref{thm.base}.

In a sense, this gives us a divide-and-conquer strategy: the task of finding a minimum length alternating path from an unmatched vertex $f$ to $v$ is divided into two tasks: find a minimum path from $f$ to $b$ and from $b$ to $v$, where $b$ is the base of $v$ and of the blossom containing $v$. The latter path is contained in the blossom, as stated above. In turn, each of these two paths may themselves pass through arbitrarily nested smaller blossoms and the process is repeated recursively.

Matching has had a long and distinguished history within graph theory and combinatorics, spanning more than a century \cite{LP.book}. Its exalted status in the theory of algorithms arises from the fact that its study has yielded quintessential paradigms and powerful algorithmic techniques  which have contributed to forming the modern theory of algorithms as we know it today. These include definitions of the classes $\cP$ \cite{Edmonds} and $\#\cP$ \cite{Val.permanent}, the primal-dual paradigm \cite{Kuh55}, the equivalence of random generation and approximate counting for self-reducible problems \cite{count.JVV}, defining facets of the convex hull of solutions to a combinatorial problem \cite{Ed.poly}, the canonical paths argument in the Markov chain Monte Carlo method \cite{JS89}, and the Isolation Lemma \cite{MVV}.

\subsection{Contributions of this Paper}
\label{sec.contributions}

Besides stating the contributions of this paper, at the end of this section, we will also state the contributions of \cite{va.matching} and point out the nature of its shortcomings due to which the current paper is called for. This section refers to several terms which are defined later in the paper. Even so, one can easily get its gist, though to fully understand it, one may have to come back to this section after reading ahead.

Let $v$ be a vertex of tenacity $t$, where $t$ is an odd number with $t_m \leq t < l_m$, where $t_m$ is the tenacity of a minimum tenacity vertex and $l_m$ is the length of a minimum length augmenting path. In order to define the base of $v$, we first need to prove Claim \ref{claim.base}. However, its proof requires the notion of a blossom and its associated properties. On the other hand, blossoms can be defined only after defining the base of a vertex. Therefore, we are faced a chicken-and-egg problem. 

Our resolution of this problem involves carrying out an induction on tenacity $t$, starting with $t = t_m$.  Once Claim \ref{claim.base} is proven for vertices of tenacity $\leq t$, the base of vertices of tenacity $\leq t$ can be defined. Following this, blossoms of tenacity $t$ can be defined and properties of these blossoms and properties of paths traversing through these blossoms can be established. In the next step of the induction, these facts are used critically  to prove Claim \ref{claim.base} for the next higher value of tenacity. 

 For proving the induction basis and step, a further induction is needed, on minlevels of vertices, i.e., the core pf the proof is structured as a double induction. We also give a recursive definition of blossoms  from the viewpoint of minimum length alternating paths. This definition is simpler than the one given in \cite{va.matching} and more appropriate for use in our inductive proof. 

The algorithm itself involves two main ideas: the new search procedure called {\em double depth first search (DDFS)} and the precise synchronization of events. The former is described in Section \ref{sec.DDFS} in a completely self-contained manner, so it can be read without reading the rest of  the paper. The latter is described in Section \ref{sec.alg}, via Figures \ref{fig.early1} and \ref{fig.early2}. For the purpose of this synchronization, the algorithm for a phase is organized in {\em search levels}. Assume that $\mn(v) = i+1$ and $\t(v) = t$. Then $\mn(v)$ is found during search level $i$ and $\mx(v)$ is found during search level $(t-1)/2$.

The following steps were taken to render the algorithm easier to comprehend.

\begin{enumerate}
	\item DDFS is described in plain English in Section \ref{sec.DDFS}; readers who prefer to understand it via a pseudocode can find it in \cite{va.matching}. 
	\item For ease of comprehension, DDFS has been described in the simpler setting of a directed, layered graph $H$. In the algorithm, DDFS is run on the original graph $G$. However, describing this process, as was done in \cite{va.matching}, is too cumbersome. Instead, we provide a mapping from $G$ to $H$ in Section \ref{sec.H} via which the reader can easily trace the steps DDFS executes in $G$. To further help the reader, in the many illustrative examples given, the distance of vertices from the unmatched vertex is proportional to their minlevels, as would be in the corresponding graph $H$.
	\item The exact working of the algorithm depends on the manner in which it resolves ties, and that determines the structures found by it. We make a very clear distinction between the structures found by the algorithm and the graph-theoretically defined structural notions,  e.g., the former include petal and bud whereas the latter include blossom and base, see Section \ref{sec.Petal-Bud}. To clarify matters further, a formal relationship between these notions is established in Lemma \ref{lem.all}. 
\end{enumerate}

We finally address the following question, ``Why is it essential to formalize such an elaborate purely graph-theoretic structure for proving correctness of the MV algorithm?'' For the ensuing discussion, we will assume that the reader is familiar with the definitions of $\mn(v)$, $\mx(v)$, $\e(v)$, $\o(v)$, tenacity and bridge. 

Assume that $\mn(v) = i+1$ and $\mx(v) = j+1$, so that $\t(v) = i + j + 2 = t$, say. It turns out that there is a neighbor, say $u$, of $v$, such that $\e(u) = i$ or $\o(u) = i$, depending on the parity of $i$. Therefore one step of breadth first search, while searching from $u$, will lead to assigning $v$ its correct minlevel. We will say that $u$ is the {\em agent that assigns $v$ its minlevel}. 

In contrast, none of the neighbors of $v$ may have $j$ as one of its levels. For instance, vertex $a$ in the graph of Figure \ref{fig.BFSH} has $\mx(a) = \o(a) = 11$; however, none of its neighbors has an evenlevel of 10. What then is the agent that assigns $v$ its maxlevel? If we do define such an agent graph-theoretically, we would then be left with the task of proving the existence of this agent for each vertex $v$. Additionally, we will need to prove that this ``agent'' is found well in time, so that $v$ is assigned its maxlevel during search level $(t-1)/2$ to ensure correct synchronization.  

This paper does provides very precise answers to these questions. The agent is a {\em bridge} whose  tenacity equals $\t(v)$. Proving the existence of such a bridge can be seen as the main outcome of the elaborate double induction mentioned above; see Statement 2 of Theorem \ref{thm.base}.  

\cite{va.matching} accurately identified the property of bridges stated above as the crux of the matter for proving correctness of the MV algorithm. At a high level, the interplay between the notions of tenacity, base, blossom and bridge was also accurately identified. However, the actual definitions and proofs were flawed; we give illustrative examples below. 

In \cite{va.matching}, the central notion of base of a vertex was defined for any vertex of finite tenacity. In the current paper, base has been defined for any vertex having tenacity in the range specified above. In Figure \ref{fig.nobase} we give an example of a vertex of tenacity $t_m$ which does not have a base. Perhaps worse, \cite{va.matching} proceeds to ``prove'' in Theorem 3 that every vertex of finite tenacity has a well-defined base. It turns out that this entire proof is incorrect, even for vertices having tenacity in the specified range, because of the chicken-and-egg problem stated above.

 In Section \ref{sec.laminar}, we prove that the set of blossoms form a laminar family and this fact plays an important role in our proof. Before embarking on proving this fact, we present an attempt at constructing a counter-example. The subtle reason due to which this counter-example fails, indicates how non-trivial the proof would be. This development went unnoticed in \cite{va.matching}, leading to an incorrect ``proof'' given in Lemma 8 in \cite{va.matching}.

\subsection{Running Time and Related Papers}
\label{sec.running}

The MV algorithm finds minimum length augmenting paths in phases; each phase finds a maximal set of disjoint such paths and 
augments the matching along all paths. $O(\sqrt{n})$ such phases suffice for finding a maximum matching \cite{Karp,Karzanov}.
The MV algorithm executes a phase in almost linear time. Its precise running time is $O(m \sqrt{n} \cdot \alpha(m, n))$ in the pointer model, 
and $O(m \sqrt{n})$ in the RAM model (see Theorem \ref{thm.time} for details). 
As is standard, $n$ denotes the number of vertices and $m$ the number of edges in the given graph.

We note that small theoretical improvements to the running time, for the case of
very dense graphs, have been given in recent years: $O(m \sqrt{n} {{\log (n^2/m)} / {\log n}})$ \cite{GKarzanov}  
and $O(n^w)$ \cite{Mucha}, where $w$ is the best exponent of $n$ for multiplication of
two $n \times n$ matrices. The former improves on MV for $m = n^{2 -o(1)}$ and the latter for $m = \omega(n^{1.85})$; additionally,
the latter algorithm involves a large multiplicative constant in its running time which comes from the use of fast matrix multiplication
as a subroutine in this algorithm.

Prior to \cite{MV}, Even and Kariv \cite{EK75} had used the idea of finding augmenting paths in phases to obtain an $O(n^{2.5})$ maximum
matching algorithm. However, the algorithm is quite complicated, there is no journal version of the result, and its correctness is hard to ascertain.

Subsequent to \cite{MV}, \cite{GTarjan2} gives an efficient scaling algorithm for finding a minimum weight matching in a general graph with integral edge weights and at the end of the paper, it claims that the unit weight version of their algorithm achieves the same running time as MV. A much better version of the weighted graph approach to cardinality matching was given recently in \cite{Gabow2017weighted}, again achieving the same running time. The rest of the history of matching algorithms is very well documented and will not be repeated here, e.g., see \cite{LP.book,va.matching}.

\section{Double Depth First Search (DDFS)}
\label{sec.DDFS}

This section is fully self-contained and describes the procedure of double depth first search (DDFS), which works by growing two DFS trees in a highly coordinated manner. For ease of comprehension, in this section we will present DDFS in the simplified setting of a directed, layered graph $H$. The setting in which it is used in the MV algorithm is more complex; we will explicitly provide a mapping between the settings to show how the ideas carry over.  

{\bf Directed, layered graph $H = (V, E)$:}
$V$ is partitioned into $h + 1$ layers, for some $h > 0$, namely $l_h, \ldots l_0$, with $l_h$ being the {\em highest layer} and $l_0$ the {\em lowest layer}. Each edge in $E$ runs from a higher to a lower layer, not necessarily consecutive. The layer number of a vertex $v$ is denoted by $l(v)$. If $l(u) < l(v)$, then we will say that $u$ {\em is deeper than} $v$. Graph $H$ contains two special vertices, $r$ and $g$, for {\em red} and {\em green}, not necessarily in the same layer. Finally, $H$ satisfies:

\noindent
{\bf DDFS Requirement:} Starting from every vertex $v \in V$, there is a path to a vertex in layer $l_0$.

Vertex $v$ will be called a {\em bottleneck} if every path from $r$ to $l_0$ and every path from $g$ to $l_0$ contains $v$; $v$ is allowed to be $r$ or $g$ or a vertex in layer $l_0$. Let $p$ be a  path from $r$ or $g$ to layer $l_0$. Since layer numbers on $p$ are monotonically decreasing,  if there is a bottleneck, the one having highest level must be unique. It will be called the {\em highest bottleneck} and we will denote it by $b$. {\em Case 1} and {\em Case 2}, respectively, will denote whether there is a bottleneck or not. In Case 2, there must be distinct vertices $r_0$ and $g_0$ in layer $l_0$ such that there are disjoint paths from $r$ to $r_0$ and $g$ to $g_0$.  

In Case 1, let $V_b$ ($E_b$) be the set of all vertices (edges) that lie on all paths from $r$ and $g$ to $b$, and in Case 2, let $E_p$ be the set of all edges that lie on all paths starting from $r$ or $g$ and ending at $r_0$ or $g_0$.

{\bf The objective of DDFS:}
The first objective of DDFS is to determine which case holds. Furthermore, in Case 1, it needs to find the highest bottleneck, $b$, and partition the vertices in $V_b - \{b\}$ into two sets $S_R$ and $S_G$, called the {\em red set} and {\em green set} respectively, with $r \in R$ and $g \in G$. These sets should satisfy: 
\begin{enumerate}
	\item There is a path from $r$ to $b$ in $S_R \cup \{b\}$ and a path from $g$ to $b$ in $S_G \cup \{b\}$.
	\item There are two spanning trees, $T_r$ and $T_g$, in $S_R \cup \{b\}$ and $S_G \cup \{b\}$, and rooted at $r$ and $g$, respectively. Furthermore, DDFS needs to find such a pair of trees.
\end{enumerate}
In Case 2, DDFS needs to find distinct vertices $r_0$ and $g_0$ in layer $l_0$, and vertex disjoint paths from $r$ to $r_0$ and $g$ to $g_0$. 

It is easy to see that once DDFS meets this objective, it provides the following certificate.

\noindent
{\bf DDFS Certificate:} In Case 1, using the trees $T_r$ and $T_g$ we get: For every vertex $v \in V_b - \{b\}$, if $v$ is red, there a path from $r$ to $v$ in $T_r$ and a disjoint path from $g$ to $b$ in $T_g$. And if $v$ is green, there a path from $g$ to $v$ in $T_g$ and a disjoint path from $r$ to $b$ in $T_r$. In Case 2, there are vertex disjoint paths from $r$ to $r_0$ and $g$ to $g_0$.

\noindent
{\bf Running time:} The running time of DDFS needs to be $O(|E_b|)$ in Case 1 and $O(|E_p|)$ in Case 2.  

{\bf The two DFSs and their coordination:} DDFS works by growing two DFS trees, $T_r$ and $T_g$, rooted at $r$ and $g$, respectively, in a highly coordinated manner. In Case 1, $b$ will be in both trees, and other than $b$, the two trees will be vertex-disjoint. The vertices in these trees, other than $b$, will be marked {\em red} and {\em green}, respectively. Initially, all vertices are marked ``unvisited'' and all edges are marked ``unexplored''.

The first time $T_r$ or $T_g$ visits an unvisited vertex, $v$, it is marked ``visited'' and its color is set accordingly. However, later on, $v$ may have to be given to the other tree and its color changed. Thus, the two trees are not grown in a greedy manner. For example, suppose $T_r$ acquires vertex $v$ which lies on all paths from $g$ to $b$ or $l_0$. Clearly, $v$ is not critical to $T_r$, since there are disjoint paths to $b$ or to $l_0$. In this case, $T_r$ should rescind $v$, allow $T_g$ to take $v$, and backtrack in order to look for an alternative path to a vertex as deep as $v$. A systematic way of carrying this out is explained below and this is the main point behind DDFS. 

Modulo one crucial aspect, which is pointed out below, both trees function as normal DFS trees in a directed graph. Each tree grows by adding unvisited vertices and each of its vertices, other than the root, has a unique {\em parent}. Each tree also has a {\em center of activity}, i.e., the vertex it is currently exploring. When it has explored all outgoing edges from its center of activity, say $v$, it backtracks to the parent of $v$. 

Let $C_r$ and $C_g$ denote the current centers of activity of $T_r$ and $T_g$, respectively; the levels of these two vertices will be denoted by $l(C_r)$ and $l(C_g)$. $C_r$ and $C_g$ are initialized to $r$ and $g$, respectively. By ``$T_r$ moves'' we mean the following. Assume $C_r = v$. Then $T_r$ will pick a previously unexplored edge out of $v$, say $(v, u)$, and if $u$ is still unvisited, it will move $C_r$ to $u$. If $u$ is already marked ``visited'', $T_r$ will seek another edge out of $v$. If all outgoing edges from $v$ are already explored, it will backtrack. 

Clearly, if we were growing only one tree, then due to the DDFS Requirement, it would find a single path all the way from its root to layer $l_0$; however, we are growing two trees in a coordinated manner, so that both of them are maximally deep. We next give the rules for this  coordination. We will adopt the (arbitrary) convention that $C_r$ will ``keep ahead of'' $C_g$ and $C_g$ will try to ``catch up''. Following this  convention, if $l(C_r) \geq l(C_g)$, then $C_r$ moves, and if $l(C_r) < l(C_g)$, then $C_g$ moves to catch up.  

{\bf When the two centers of activity meet:} 
Finally, we state the most novel and crucial aspect of the coordination, namely the steps to be taken if both centers of activity meet at a vertex, i.e., $C_r = C_g = v$, say. Observe that because edges can be long, i.e., not necessarily going to the next lower layer, either of the trees may have gotten to $v$ first. DDFS needs to determine if $v$ is the highest bottleneck, and if not, then which of the trees can find an alternative path at least as deep as $v$, so search may resume.  

We will adopt the convention that $v$ is first given to $T_g$, and $T_r$ will attempt to find an alternative path. For this purpose, it backtracks from $v$ and continues searching. If it succeeds in finding a path to a vertex which is at least as deep as $v$, DDFS resumes. However, if it backtracks all the way to $r$, then we will allocate $v$ to $T_r$ and let $T_g$ find an alternative path after backtracking from $v$. 

At this point, $T_r$ also needs to update a pointer called Barrier. The purpose of Barrier is to prevent $T_r$ from backtracking from a vertex more than once. At the start of DDFS, Barrier is initialized to $r$. Observe that at the current stage in the procedure, $T_r$ has backtracked from all its vertices in layers $l(r)$ to $l(v)$ and from now on, $T_r$ needs to confine itself to layers lower than $l(v)$. Therefore,  Barrier is updated to $v$. If in the future the two centers of activity meet again, say at $u$, and $T_r$ backtracks all the way to Barrier, then it will not backtrack any further. Moreover, it will update Barrier to $u$.

Next consider the case that after backtracking from $v$, $T_g$ finds an alternative path to a vertex at least as deep as $v$. Then DDFS resumes and it is $T_r$'s turn to move. On the other hand, if $T_g$ backtracks all the way to $g$ without finding another path as deep as $v$, then $v$ is declared the highest bottleneck. If so, DDFS terminates and does not search any vertices below $v$.  

If both centers of activity reach distinct vertices at layer $l_0$, say $r_0$ and $g_0$, then DDFS terminates in Case 2.

\begin{theorem}
\label{thm.DDFS}
DDFS accomplishes the objectives stated above in the required time.
\end{theorem}

\begin{proof}
In Case 1, tree $T_r$ contains paths consisting of red colored vertices from $r$ to $b$ and from $r$  to each red vertex. A similar claim holds about tree $T_g$. In Case 2, there is a path consisting of red colored vertices from $r$ to $r_0$ in tree $T_r$. An analogous statement holds about $T_g$. Therefore, the DDFS Certificate holds. 

	Finally, it is easy to see that each edge of $H$ is explored by at most one tree and if so, only once. Because of the Barrier, each tree backtracks from each vertex at most once. The theorem follows.
\end{proof}

\section{Basic Definitions}
\label{sec.basic-definitions}

A matching $M$ in an undirected graph $G = (V, E)$ is a set of edges no two of which meet at a vertex. Our problem is to find a matching of 
maximum cardinality in the given graph. All definitions henceforth are w.r.t. a fixed matching $M$ in $G$. 
Edges in $M$ will be said to be {\em matched} and those in $E-M$ will be said to be {\em unmatched}. Vertex $v$ will be said to be 
matched if it has a matched edge incident at it and unmatched otherwise.

An {\em alternating path} is a simple path whose edges alternate between $M$ and $E-M$, i.e., 
matched and unmatched. An alternating path that starts and ends at unmatched vertices is called an {\em augmenting path}. Clearly the number
of unmatched edges on such a path exceeds the number of matched edges on it by one\footnote{Observe that if $M = \emptyset$, then any edge is 
an augmenting path, of length one.}.
Its significance lies in that flipping matched and
unmatched edges on such a path leads to a valid matching of one higher cardinality. Edmonds' matching algorithm operates by iteratively finding
an augmenting path w.r.t. the current matching, which initially is assumed to be empty, and augmenting the matching. When there are no 
more augmenting paths w.r.t. the current matching, it can be shown to be maximum.

The MV algorithm finds augmenting paths in phases as proposed in \cite{Karp,Karzanov}. In each {\em phase}, it finds a maximal set of disjoint 
minimum length augmenting paths w.r.t. the current matching and it augments along all paths. \cite{Karp,Karzanov} show that only $O(\sqrt{n})$ 
such phases suffice for finding a maximum matching in general graphs. The remaining task is designing an efficient algorithm for a phase.


\notation{(Length of Minimum Length Augmenting Path)}
Throughout, $l_m$ will denote the length of a minimum length augmenting path in $G$; if $G$ has no augmenting paths, we will assume that $l_m = \infty$.


\definition{(Evenlevel and oddlevel of vertices)}
The evenlevel (oddlevel) of a vertex $v$, denoted $\e(v)$ ($\o(v)$), is defined to be the length of a
minimum even (odd) length alternating path from an unmatched vertex to $v$; moreover, each such path will be called
an $\e(v)$ ($\o(v)$) path. If there is no such path, $\e(v)$ ($\o(v)$) is defined to be $\infty$. 

We will typically denote an unmatched vertex by $f$. Its evenlevel is zero and its oddlevel is the length of the shortest augmenting path starting at $f$; if no augmenting path starts at $f$, $\o(f) = \infty$. The length of a minimum length augmenting paths w.r.t. $M$ is the smallest oddlevel of an unmatched vertex. In all the figures, matched edges are drawn dotted, unmatched edges solid, and unmatched vertices are drawn with a small circle.

\begin{figure}[ht]
\begin{minipage}[b]{0.5\linewidth}
\centering
\includegraphics[width=\textwidth]{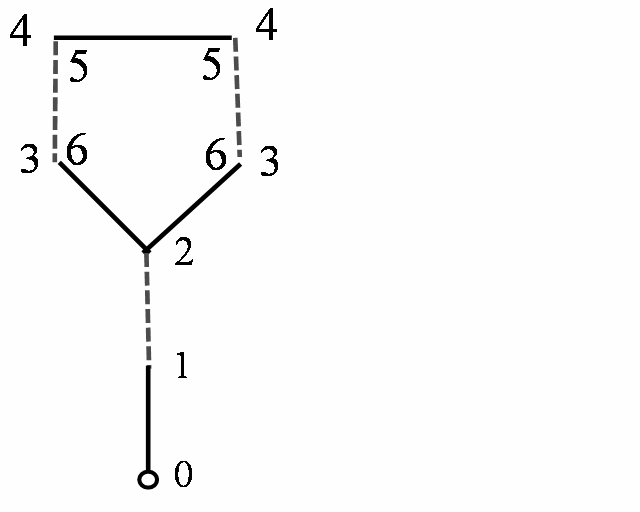}
\caption{
Evenlevels and oddlevels of vertices are indicated; missing levels are $\infty$.}
\label{fig.vten}
\end{minipage}
\hspace{0.5cm}
\begin{minipage}[b]{0.5\linewidth}
\centering
\includegraphics[width=\textwidth]{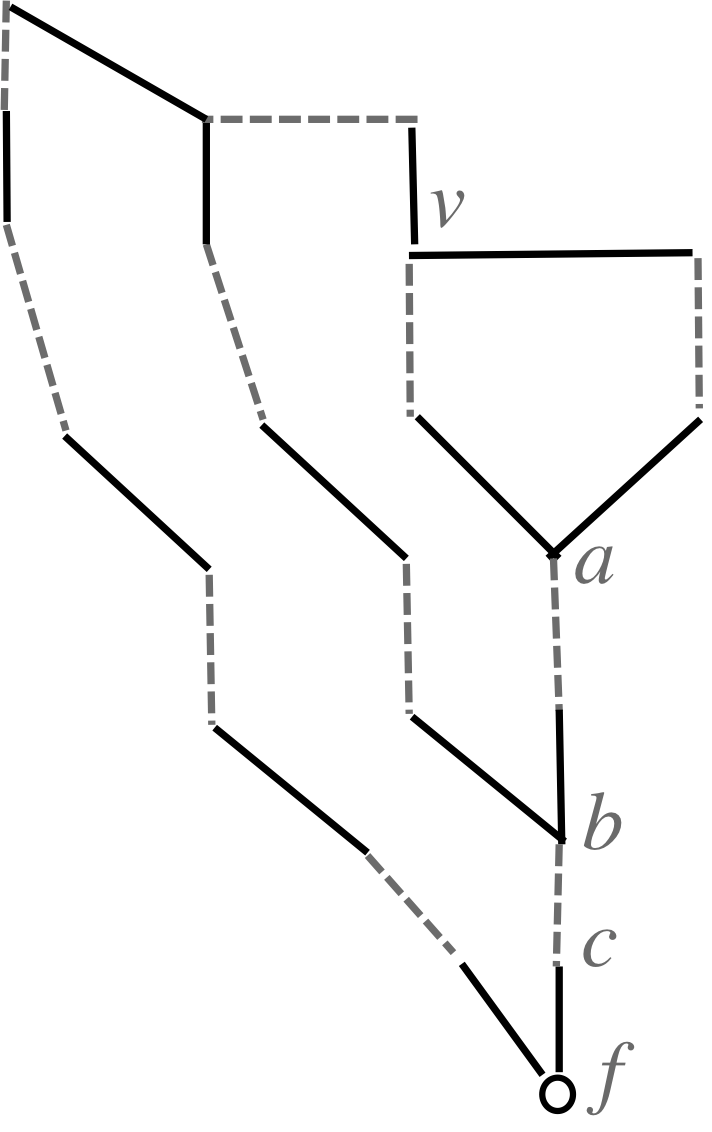}
\caption{Vertex $v$ is not BFS-honest on $\o(a)$ and $\e(c)$ paths.}
\label{fig.BFSH}
\end{minipage}
\end{figure}

\definition{(Maxlevel and minlevel of vertices)}
For a vertex $v$ such that at least one of $\e(v)$ and $\o(v)$ is finite, $\mx(v)$ ($\mn(v)$) is defined to be the bigger (smaller) of the two.

\definition{(Outer and inner vertices)}
A vertex $v$ with finite minlevel is said to be {\em outer} if $\e(v) < \o(v)$ and {\em inner} otherwise.

Let $p$ is an alternating path from unmatched vertex $f$ to $v$ and let $u$ lie on $p$. Then 
$p[f \ \mbox{to} \ u$ will denote the part of $p$ from $f$ to $u$. Similarly 
$p[f \ \mbox{to} \ u)$ denotes the part of $p$ from $f$ to the vertex just before $u$, etc.

Breadth first search (BFS) is the natural way of finding minimum length paths in an undirected,  unweighted graph. For finding minimum length {\em augmenting paths} in bipartite graphs, a slight extension to an {\em alternating breadth first search} suffices, e.g., see Section \ref{sec.alg} (for a complete description, see Section 2.1 in \cite{va.matching}). The property of bipartite graphs due to which this method succeeds is captured in the next definition. It implies that if $p$ is a minimum length alternating path from $f$ to $u$, then appending to it the shortest path from $u$ to $v$, respecting the alternating nature of the resulting path, suffice to find a minimum length alternating path from $f$ to $v$.

\definition{(Breadth first search honesty)} 
\label{def.Honesty}
If $p$ is a minimum length alternating path from $f$ to $v$ and $u$ lies on $p$ then $p[f \ \mbox{to} \ u]$ is a minimum length alternating path from $f$ to $u$. 

This elementary property does not hold in non-bipartite graphs, e.g., in Figure \ref{fig.BFSH}, $\o(v) = 7$. However, on the $\o(a)$, $\o(b)$ paths, $v$ occurs at a length of 9 and 11, respectively, i.e., $v$ is not BFS-honest w.r.t. these paths. This implies that we need to find longer and longer oddlevel paths to $v$ in order to find minimum length alternating paths to other vertices, namely $a$ and $b$ in this case. Informally, finding short paths is easy and long paths is hard e.g. Hamiltonian path. Indeed, the problem being attacked may seem intractable at first sight.


\begin{figure}[ht]
\begin{minipage}[b]{0.5\linewidth}
\centering
\includegraphics[width=\textwidth]{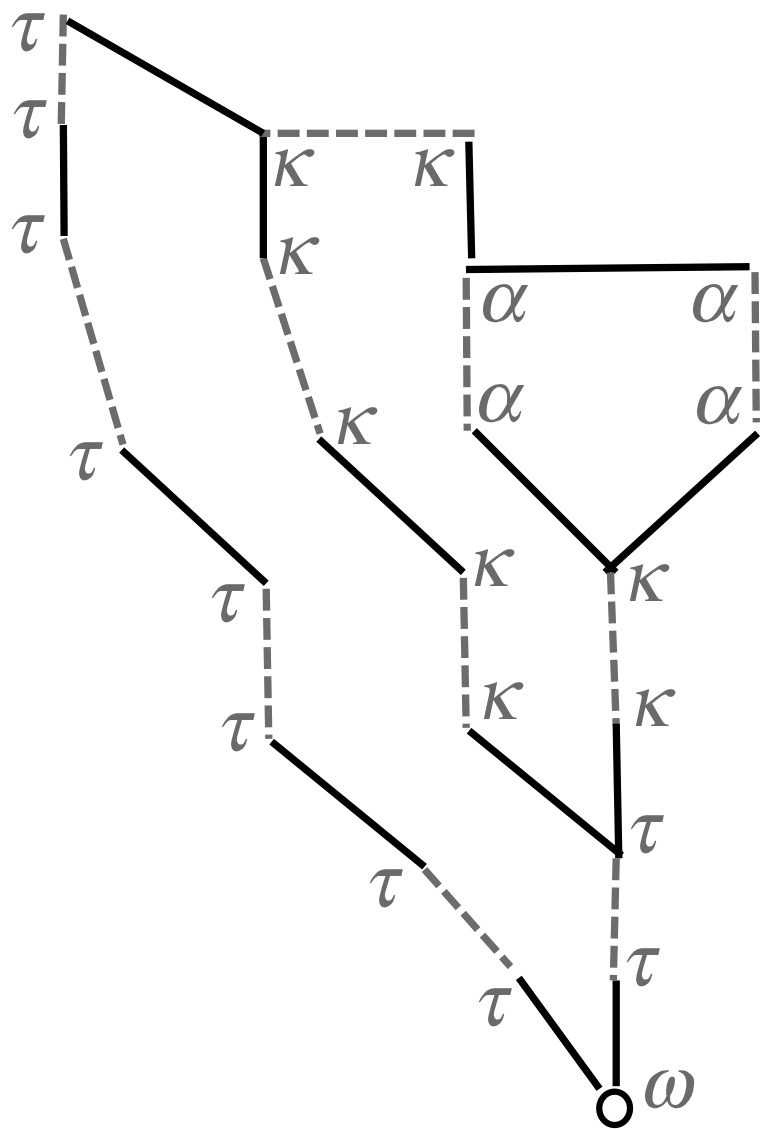}
\caption{The tenacity of vertices is indicated; here $\alpha = 13, \ \kappa = 15$, $\tau = 17$ and $\omega = \infty$.}
\label{fig.verten}
\end{minipage}
\hspace{0.5cm}
\begin{minipage}[b]{0.5\linewidth}
\centering
\includegraphics[width=\textwidth]{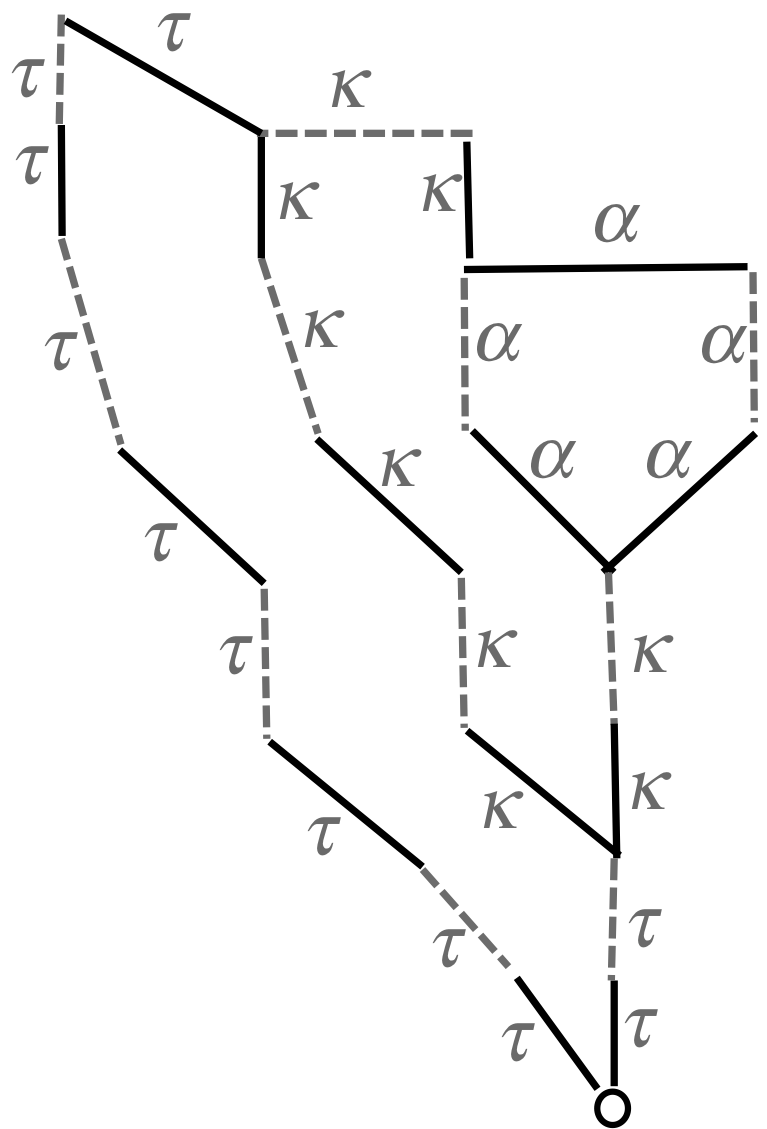}
\caption{The tenacity of each edge is indicated; here $\alpha = 13, \ \kappa = 15$ and $\tau = 17$.}
\label{fig.edgeten}
\end{minipage}
\end{figure}


\definition{(Tenacity of vertices and edges)}
\label{ref.tenacity}
Define the tenacity of vertex $v$, $\t(v) = \e(v) + \o(v)$.
If $(u, v)$ is an unmatched edge, then $\t(u, v) = \e(u) + \e(v) + 1$, and if it is matched, $\t(u, v) = \o(u) + \o(v) + 1$.

\notation{(Minimum tenacity of a vertex in $G$)}
Throughout, $t_m$ will denote the tenacity of a minimum tenacity vertex in $G$. 

The examples in Figures \ref{fig.vten}, \ref{fig.verten} and \ref{fig.edgeten} illustrate these notions. The notion of tenacity is central to the structural facts that follow. 

\definition{(Predecessor, prop and bridge)}
Consider a $\mn(v)$ path and let $(u, v)$ be the last edge on it; clearly, $(u, v)$ is matched if $v$ is outer and unmatched otherwise.
In either case, we will say that $u$ is a predecessor of $v$ and that edge $(u, v)$ is a prop.
An edge that is not a prop will be defined to be a bridge.

In Figure \ref{fig.BFSH}, the two horizontal edges and the oblique unmatched edge at the top 
are bridges and the rest of the edges of this graph are props. In Figure \ref{fig.vb}, $(w, w')$ and $(v, v')$ are bridges, and
in Figure \ref{fig.anamoly}, $(w, w')$ and $(u, v)$ are bridges; the rest of the edges in these two graphs are props.


\begin{figure}[ht]
\begin{minipage}[b]{0.4\linewidth}
\centering
\includegraphics[width=\textwidth]{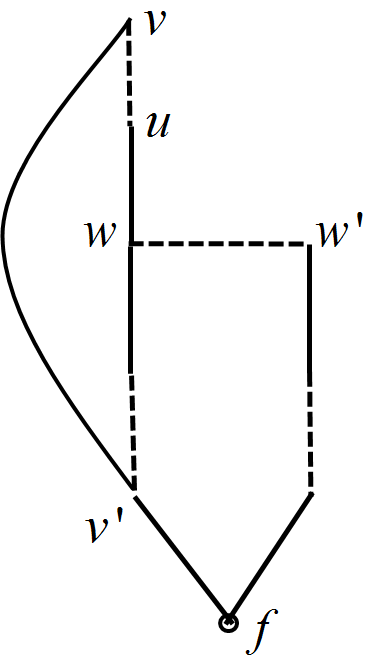}
\caption{Edges $(w, w')$ and $(v, v')$ are bridges.}
\label{fig.vb}
\end{minipage}
\hspace{3.2cm}
\begin{minipage}[b]{0.4\linewidth}
\centering
\includegraphics[width=\textwidth]{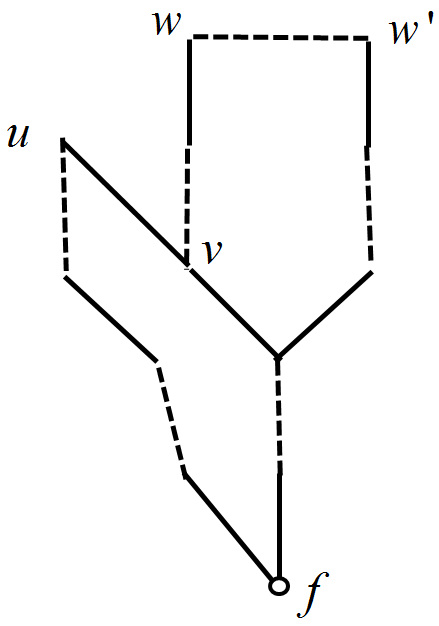}
\caption{Edges $(w, w')$ and $(u, v)$ are bridges.}
\label{fig.anamoly}
\end{minipage}
\end{figure}


\definition{(The support of a bridge)}
Let $(u, v)$ be a bridge of tenacity $t \leq l_m$. Then, its support is defined to be $\{w ~|~ \t(w) = t \ \mbox{and} \ \exists \ \mbox{a} \
\mx(w) \ \mbox{path containing} \ (u, v) \}$.

In the graph of Figures \ref{fig.BFSH}, \ref{fig.verten} and \ref{fig.edgeten}, the supports of the bridges of tenacity $\alpha, \kappa$
and $\tau$ are the set of vertices of tenacity $\alpha, \kappa$ and $\tau$, respectively. In the graph of Figure \ref{fig.vb},
edges $(w, w')$ and $(v, v')$ are bridges of tenacity 7 and 13, respectively. Unmatched vertex $f$ is not in the support of either bridge.
The support of bridge $(v, v')$ is $\{u, v\}$ and the support of bridge $(w, w')$ is all the remaining vertices other than $f$.

\section{The Algorithm}
\label{sec.alg}

The MV algorithm starts with the empty matching and executes iterations which we will call {\em phases} until there are no augmenting paths w.r.t. the current matching, i.e., the current matching is maximum. Each phase finds a maximal set of disjoint minimum length augmenting paths paths and augments the matching along these paths. Each phase is itself iterative and in each iteration, the algorithm calls the procedures MIN and MAX, which find minlevels and maxlevels of vertices, respectively.

\subsection{Procedures MIN and MAX}
\label{sec.sync}

At the beginning of a phase, all unmatched vertices are assigned a minlevel of 0, the rest are assigned a temporary minlevel of $\infty$. No vertices are assigned maxlevels at this stage. The algorithm for a phase is organized by search levels; in each search level, MIN executes one step of alternating BFS and is followed by MAX.  

If $i$ is even (odd), MIN searches from all vertices, $u$, having an evenlevel (oddlevel) of $i$ along incident unmatched (matched) edges, say $(u, v)$. If edge $(u, v)$ has not been scanned before, MIN will determine if it is a prop or a bridge as follows.  If $v$ has already been assigned a minlevel of at most $i$, then $(u, v)$ is a bridge. Otherwise, $v$ is assigned a minlevel of $i + 1$, $u$ is declared a predecessor of $v$ and edge $(u, v)$ is declared a prop. Note that if $i$ is odd, $v$ will have only one predecessor -- its matched neighbor, and if $i$ is even, $v$ will have one or more predecessors. Once an edge is identified as a bridge, if MIN is able to ascertain its tenacity, say $t$, then the edge is inserted in the {\em list of bridges of tenacity $t$}, $Br(t)$. MIN is able to ascertain the tenacity of a bridge in all but one case; in the last case, MAX finds the tenacity of this bridge as  described below.

After MIN is done, procedure MAX uses the procedure DDFS to find all vertices, $v$, having $\t(v) = 2i+1$ and assigns these vertices their maxlevels. Their minlevels are at most $i$ and are already known and hence $\mx(v) =  2i+1 - \mn(v)$ can be computed. Let $l_m$ be the length of a minimum length augmenting path in a phase. Then during search level $j_m$, where $l_m = 2 j_m + 1$, a maximal set of such paths is found. See Algorithm \ref{alg} for a summary of the main steps.

The MV algorithm runs DDFS not on a directed, layered graph $H$, as described in Section \ref{sec.DDFS}, but on the original graph $G$. The precise mapping from $G$ to $H$ is given in Section \ref{sec.H}. Using this mapping, one can trace back the steps taken in $H$ onto $G$, thereby obtaining an algorithm that works entirely on $G$, without ever constructing $H$. Indeed, that is the right way to program the algorithm. However, the mapping gives a simpler and clearer conceptual picture.



\bigskip

\noindent

\fbox{
\begin{algorithm}{\label{alg} \ \ \ \ \ \ \ \ \ \ \ \   At search level $i$:}

\bigskip

\step
\label{step1}
{\bf MIN:}  \\

{\bf For} each level $i$ vertex, $u$, search along appropriate parity edges incident at $u$. 

\begin{description}
\item

{\bf For} each such edge $(u, v)$, if $(u, v)$ has not been scanned before {\bf then}

\begin{description}
\item
{\bf If} $\mn(v) \geq i+1$ {\bf then} \\

\begin{description}
\item
$\mn(v) \la i+1$

\item
Insert $u$ in the list of predecessors of $v$.

\item
Declare edge $(u, v)$ a prop.

\end{description}

\item
{\bf Else} declare $(u, v)$ a bridge, and if $\t(u, v)$ is known, \\
insert $(u, v)$ in $Br(\t(u, v))$. 

\end{description}

\item
{\bf End}

\end{description}

{\bf End}

\bigskip

\step
\label{step2}
{\bf MAX:}  \\

{\bf For} each edge in $Br(2i+1)$:

\begin{description}
\item
Find its support using DDFS. 

\item
{\bf For} each vertex $v$ in the support: \\  


\begin{description}

\item
$\mx(v) \la 2i+1 - \mn(v)$

\item
{\bf If} $v$ is an inner vertex, {\bf then}

\begin{description}
\item
{\bf For} each edge $e$ incident at $v$ which is not prop, if its tenacity is known, \\ 
insert $e$ in $Br(\t(e))$.

\end{description}

{\bf End}

\end{description}

{\bf End}

\end{description}

{\bf End}



\end{algorithm}
}

\bigskip


We next point out some salient features of MIN and MAX via the graph in Figure \ref{fig.anamoly}. At search level 4, MIN searches from vertex $u$ along edge $(u, v)$ and realizes that $v$ already has a minlevel of 3 assigned to it. Moreover,
$u$ got its minlevel from its matched neighbor. Therefore, MIN correctly identifies edge $(u, v)$ to be a bridge. However, it is not able to 
ascertain $\t(u, v)$ since $\e(v)$ is not known at this time. At search level 5, after conducting DDFS on bridge $(w, w')$ (of tenacity 11),
MAX will assign $\mx(v) = 8$, which is also $\e(v)$. Therefore, at that time, $\t(u, v)$ can be ascertained to be 13, and edge $(u, v)$ is inserted in
$Br(13)$.

To summarize, this case happens if $(u, v)$ is an unmatched bridge such that the evenlevel of one of the endpoints, say $v$, has not been determined at the point when MIN realizes that $(u, v)$ is a bridge; if so $v$ is an inner vertex. The evenlevel of $v$ will be determined by MAX at search level $(\t(v) -1)/2$ and at this point, $\t(u, v)$ is ascertained and the edge is inserted in $Br(\t(u, v))$. For each bridge in the first five figures, its tenacity gets ascertained by MIN (including the bridge $(v, v')$ in Figure \ref{fig.vb}).

An important point to note in Figure \ref{fig.anamoly}, is that $\t(v) < \t(u, v)$. This ensures that $\mx(v)$ is known at search level $(\t(v) -1)/2$, i.e., before the search level at which bridge $(u, v)$ needs to be processed by MAX, namely search level $(\t(u, v) -1)/2$. 

Task 2 in Theorem \ref{thm.levels} proves that by the end of execution of procedure MIN at search level $i$, the algorithm would have identified every bridge of tenacity $2i + 1$. At this point, procedure MAX gets executed and it uses DDFS to find the support of each of these bridges. This yields all vertices of tenacity $2i+1$, and their maxlevels are ascertained. If such a vertex, $v$, is inner and has an incident unmatched say $(u, v)$ which is not a prop, then its tenacity is ascertained and it is inserted in $Br(\t(u, v))$.

\begin{figure}[ht]
\begin{minipage}[b]{0.5\linewidth}
\centering
\includegraphics[width=\textwidth]{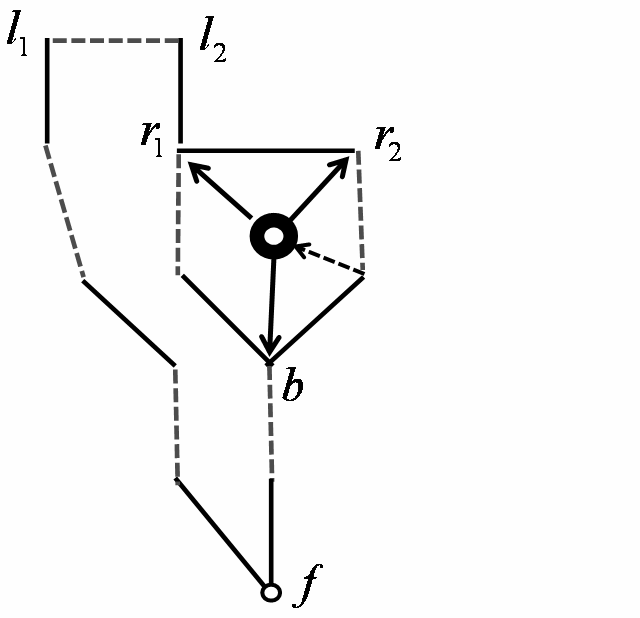}
\caption{A new petal-node is created after DDFS on bridge $(r_1, r_2)$.}
\label{fig.DDFSG}
\end{minipage}
\hspace{0.5cm}
\begin{minipage}[b]{0.5\linewidth}
\centering
\includegraphics[width=\textwidth]{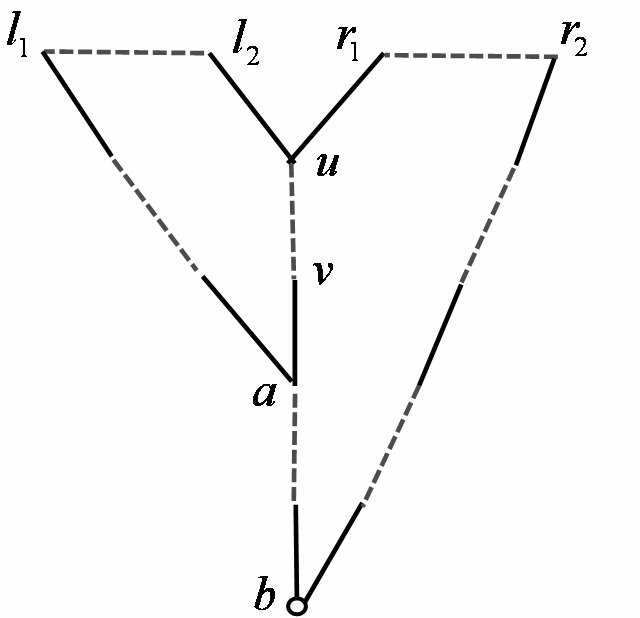}
\caption{Vertices $u$ and $v$ are in the support of both bridges of tenacity of 11.}
\label{fig.disjoint}
\end{minipage}
\end{figure}

\subsection{The Notions of Petal and Bud}
\label{sec.Petal-Bud}

Assume that DDFS is called with a bridge $(u, v)$ of tenacity $t$ and it terminates in Case 1. Then  the highest bottleneck found is called {\em bud}. The vertices of tenacity $t$ encountered by DDFS, which must lie in the support of $(u, v)$, form a new {\em petal}, which consists of all vertices in the support of $(u, v)$ minus the supports of all bridges processed thus far in this search level (which will all be of tenacity $t$). Clearly a vertex is included in at most one petal. 

In the graph of Figure \ref{fig.DDFSG}, MAX will call DDFS with the bridge $(r_1, r_2)$, which is of tenacity 9, at search level 4. The roots of the two DFSs are $r_1$ and $r_2$, and in this graph, a step of either DFS tree is to move to the predecessor of the current center of activity. Clearly, DDFS will terminate in Case 1 with $b$ as the highest bottleneck. The four vertices which constitute the support of bridge $(r_1, r_2)$ form the new petal and $b$ is the new bud found. Observe that $b$ does not belong to this petal.

To form a new petal, the algorithm executes the following steps: It creates a new node, called {\em petal-node}; this has the shape of a doughnut in Figure \ref{fig.DDFSG}. All vertices of the new petal point\footnote{To avoid cluttering Figure \ref{fig.DDFSG}, only one vertex is pointing to the petal-node.} to the petal-node; $b$ is not in the petal and does not point to the petal-node. The new petal-node points to the two endpoints of its bridge, $r_1$ and $r_2$, and to its bud, $b$. These pointers will enable the algorithm to:
\begin{enumerate}
	\item skip over petals in future DDFSs, and
	\item efficiently find an alternating path through a petal in case it goes through the petal. 
\end{enumerate}

\begin{figure}[ht]
\begin{minipage}[b]{0.5\linewidth}
\centering
\includegraphics[width=\textwidth]{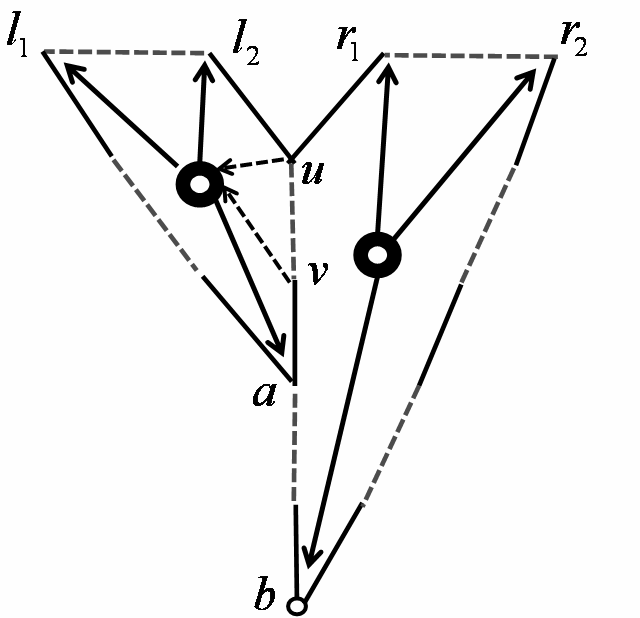}
\caption{DDFS is performed on the left bridge first, then right.}
\label{fig.left}
\end{minipage}
\hspace{0.5cm}
\begin{minipage}[b]{0.5\linewidth}
\centering
\includegraphics[width=\textwidth]{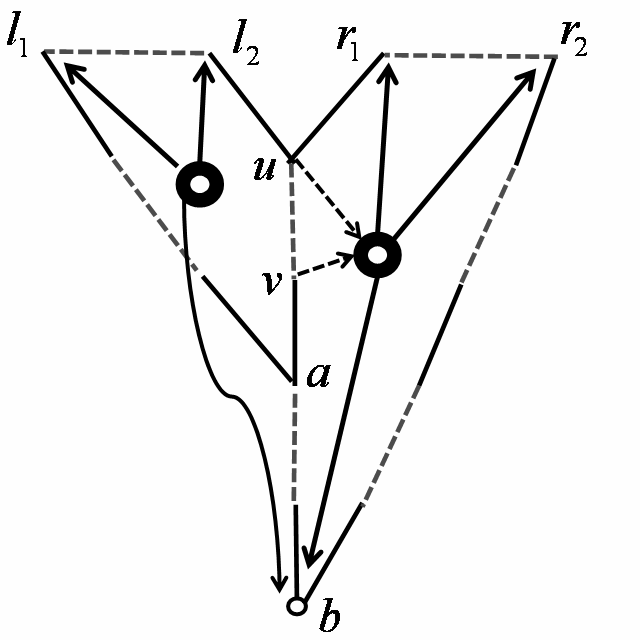}
\caption{DDFS is performed on the right bridge first, then left.}
\label{fig.right}
\end{minipage}
\end{figure}

\definition{(The bud of a vertex)}
If vertex $v$ is in a petal and the bud of this petal is $b$ then $\bd(v) = b$, and if $v$ is not in a petal, then $\bd(v) = v$. The function $\bds(v)$ is defined recursively as follows: If $\bd(v) = v$ then $\bds(v) = v$, else $\bds(v) = \bds(\bd(v))$. The bud of a petal is always an outer vertex. 

The notions of petal and bud are intimately related to the notions of blossom and base. Whereas the first pair is algorithmic --- the exact petals and buds found depend on the manner in which the algorithm resolves choices --- the second pair is purely graph-theoretic. The relationship between these notions is established in Lemma \ref{lem.all}. Here we simply note that a blossom is a union of petals and the base of a vertex $v$ will be $\bds(v)$ at the end of MAX in search level $(t-1)/2$ where $\t(v) = t$.

\subsection{The Mapping from Graph $G$ to $H$}
\label{sec.H}
 
The graph $H$ is defined each time DDFS is called. It is a function of the bridge which triggers the current DDFS and the petals which have been found at all search levels so far. The vertices of $H$ correspond to a subset of the vertices of $G$ as defined below. Assume that $v_H$ belongs to $H$ and corresponds to $v$ in $G$. Then the level of $v_H$ is defined to be $\mn(v)$. 

Assume DDFS is called with bridge $(r, g)$. Then $H$ has the two vertices $\bds(r)$ and $\bds(g)$\footnote{It is possible that $\bds(r) = \bds(g)$. This happens if the bridge $(r, g)$ has non-empty support and therefore a new petal is not formed. In this case DDFS simply aborts. An example of this phenomenon is mentioned in Section \ref{sec.examples-DDFS} in the graph of Figure \ref{fig.ten19}.}. The rest of $H$ is recursively defined as follows. If $\mn(v) = 0$ then $v$ has no predecessors in $G$, $l(v_H) = 0$ and $v_H$ has no outgoing edges. If $\mn(v) > 0$ then corresponding to each predecessor $u$ of $v$ in $G$, $H$ has the vertex $\bds(u)$ and the directed edge $(v, \bds(u))$. It is easy to confirm that $H$ satisfies the DDFS Requirement.

\subsection{Examples to Illustrate DDFS}
\label{sec.examples-DDFS}

Suppose DDFS is called with bridge $(l_1, l_2)$ in the graph of Figure \ref{fig.DDFSG}. The tenacity of this bridge is 11 and DDFS will be performed on it at search level 5. Notice that in this case, $H$ will have the edge $(l_2, b)$ and not $(l_2, r_1)$. DDFS will end in Case 1 with bottleneck $f$. The new petal is precisely the support of bridge $(l_1, l_2)$ and consists of the eight vertices of tenacity 11 in Figure \ref{fig.DDFSG}, which includes $b$. Once again, a new petal-node is created and these eight vertices point to it. In addition, the petal-node points to $l_1, l_2$ and to $f$.

Next consider the graph of Figure \ref{fig.disjoint} which has two bridges of tenacity 11, $(l_1, l_2)$ and $(r_1, r_2)$. Observe that vertices $u$ and $v$ are in the support of both these bridges. Hence, the support of bridges need not be disjoint. MAX will perform DDFS on these two bridges in arbitrary order at search level 5. 

Figure \ref{fig.left} shows the result of performing DDFS on $(l_1, l_2)$ before $(r_1, r_2)$. The first DDFS will end with bottleneck $a$. The new petal is precisely the support of $(l_1, l_2)$, consisting of six vertices of tenacity 11, including $u$ and $v$. Observe that when the second DDFS is performed, on bridge $(r_1, r_2)$, the edge out of $r_1$ in $H$ is $(r_1, a)$ and not $(r_1, u)$. This DDFS will end with bottleneck $b$ and the new petal is precisely the difference of supports of $(r_1, r_2)$ and $(l_1, l_2)$, i.e., the remaining eight vertices of tenacity 11, including $a$. 

Figure \ref{fig.right} shows the result of performing DDFS on $(r_1, r_2)$ before $(l_1, l_2)$. The first petal is the support of $(r_1, r_2)$, i.e., 10 vertices of tenacity 11, including $a$, $u$ and $v$. The second petal is the difference of supports of $(l_1, l_2)$ and $(r_1, r_2)$, i.e., 4 vertices of tenacity 11. 

In the graph of Figure \ref{fig.2paths}, DDFS called with bridge $(u, v)$ ends in Case 2, i.e., it finds two disjoint paths, indicating the presence of an augmenting path.

\begin{figure}[ht]
\begin{minipage}[b]{0.5\linewidth}
\centering
\includegraphics[width=\textwidth]{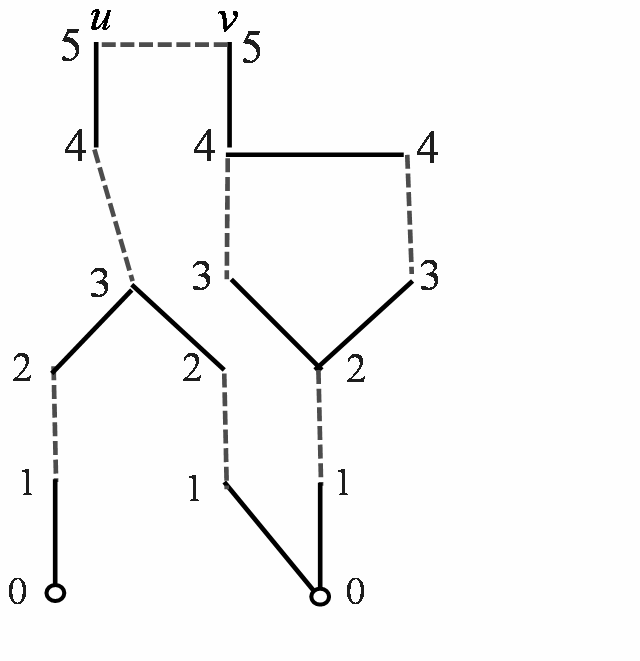}
\caption{DDFS on the bridge of tenacity 11 ends with the two unmatched vertices.}
\label{fig.2paths}
\end{minipage}
\hspace{0.5cm}
\begin{minipage}[b]{0.5\linewidth}
\centering
\includegraphics[width=\textwidth]{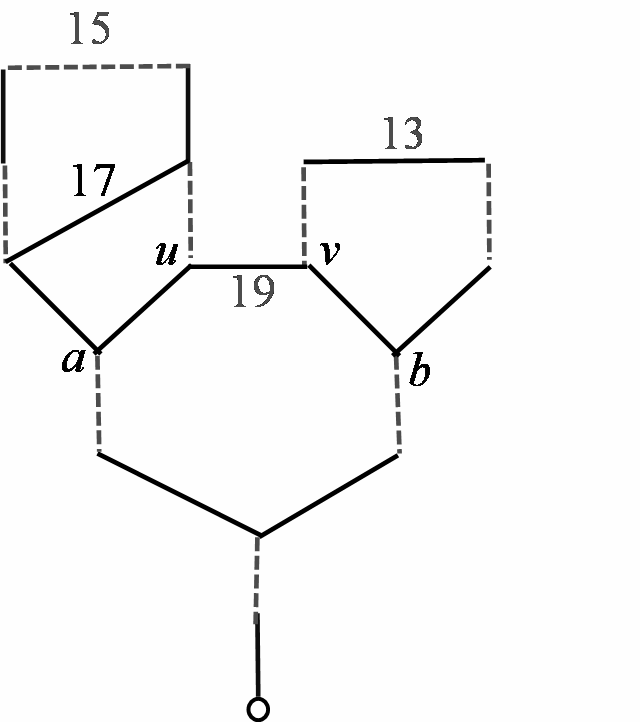}
\caption{DDFS performed on bridge $(u, v)$ starts the two DFSs at $a$ and $b$, respectively.}
\label{fig.ten19}
\end{minipage}
\end{figure}

In Figure \ref{fig.ten19} consider the situation when DDFS is called at search level 9, with the bridge $(u, v)$, which is of tenacity 19. At that point in the algorithm, the bridges of tenacity 15 and 13 would already be processed and $u$ and $v$ will already be in petals. Therefore, the roots of the two DFSs in $H$ will be $\bds(u) = a$ and $\bds(v) = b$, and not  $u$ and $v$.

All bridges considered so far had non-empty supports; however, this will not be the case in a typical graph, e.g., consider the edge of tenacity 17 in Figure \ref{fig.ten19}. Clearly the support of this bridge is $\emptyset$. DDFS will discover this right away since the $\bds$ of both endpoints of this bridge is $a$. Note that DDFS needs to be called with all bridges, even those with empty support, since this information is not available a priori.


\begin{figure}[h]
\begin{center}
\includegraphics[scale = 0.4]{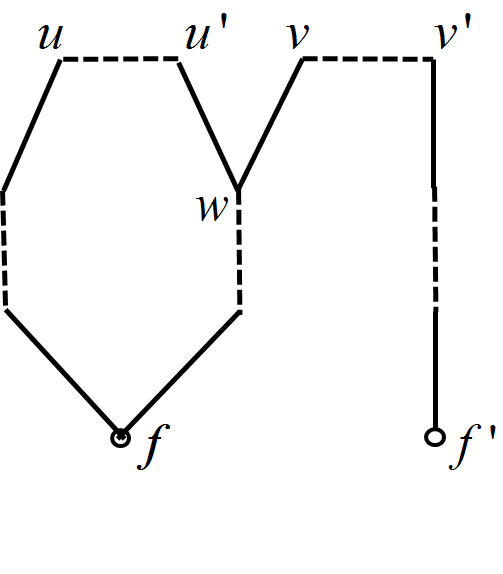}
\caption{$\t(u, u') = \t(v, v') = 7 = l_m$.}
\label{fig.two-bridges}
\end{center}
\end{figure}


\subsection{Finding the Augmenting Paths}
\label{sec.finding}

 MAX will find augmenting paths during search level $j_m$, where $l_m = 2 j_m + 1$ and $l_m$ is the length of a minimum length augmenting path in the current phase. However, not every bridge of tenacity $l_m$ leads to an augmenting path, e.g., in Figure \ref{fig.two-bridges}, suppose DDFS is called with bridge $(u, u')$ before bridge $(v, v')$. The first DDFS ends in Case 1, with bottleneck $f$. The second DDFS proceeds as follows. The DFS tree rooted at $v$ goes to $w$, since it is a predecessor of $v$. Since $w$ is in a petal, DFS must jump to $\bds(w) = f$. The second DFS tree, rooted at $v'$, follows predecessors and eventually reaches $f'$. Hence this DDFS terminates in Case 2, since each tree has found an unmatched vertex at level 0. This indicates the presence of an minimum length augmenting from $f$ to $f'$. Such a path  will be found using the procedure given in Section \ref{sec.one}.


\begin{figure}[h]
\begin{center}
\includegraphics[scale = 0.4]{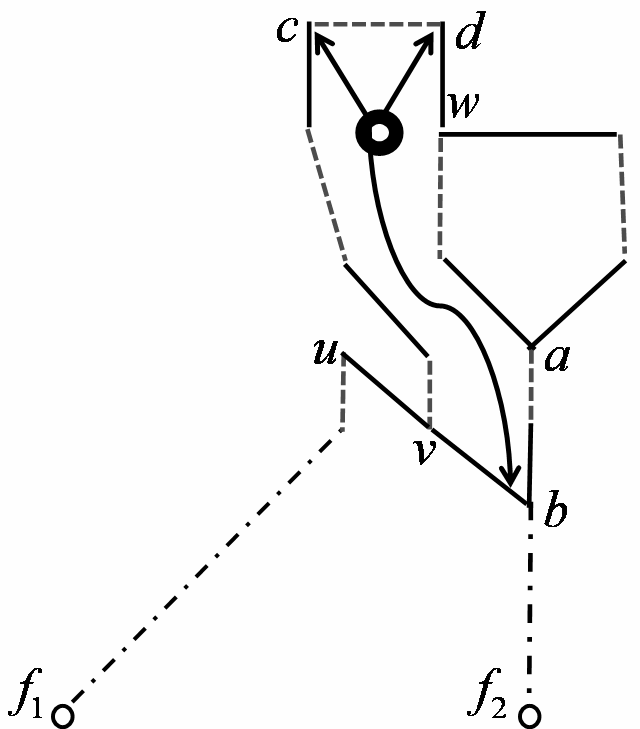}
\caption{Constructing a minimum length augmenting path between unmatched vertices $f_1$ and $f_2$.}
\label{fig.findpath}
\end{center}
\end{figure}


\subsubsection{Finding one augmenting path}
\label{sec.one}

In Figure \ref{fig.findpath}, $\mn(u) > \mn(v)$ and therefore edge $(u, v)$ is a bridge. When DDFS is performed on this bridge, assume that the red DFS trees has root $u$. Since $v$ is already in a petal with $\bds(v) = b$, the green DFS tree will have root $b$. The two trees will simply follow predecessors and will terminate at $f_1$ and $f_2$, respectively.  

A DFS from $u$ in the red tree will yield a path from $u$ to $f_1$, say $p_1$. Since $v$ is in a petal with $\bds(v) = b$, the algorithm needs to find an $\e(b; v)$ path, say $p_2$, in the petal of $v$, and a path, say $p_3$, from $b$ to $f_2$ in the green tree of the DDFS performed on bridge $(c, d)$. Then, the complete augmenting path from $f_1$ to $f_2$ will be $p_1^{-1} \bullet (u, v) \bullet p_2^{-1} \bullet p_3$. Clearly, $p_1$ and $p_3$ are easy to find. 

We next describe how to find $p_2$. The algorithm observes that $\e(v) = \mx(v)$ and therefore the $\e(b; v)$ path must use the bridge of the petal containing $v$. Using the petal node, the algorithm finds the endpoints of this bridge, namely $c$ and $d$. It notices that $c$ and $v$ have the same color, say red. Therefore, it looks for a path from $c$ to $v$ in the red tree and a path from $d$ to $b$ in the green tree. For finding the latter path, it jumps from $w$ to $\bds(w) = a$ and then follows predecessors in the green tree till it reached $b$. 

To find the complete path from $d$ to $b$ it must find an $\e(a; w)$ path in the smaller petal. This time, it observes that $\e(w) = \mn(w)$ and therefore the $\e(a; w)$ path does not use the bridge of the smaller petal. Instead it is found by doing a DFS in the green tree, assuming that the color of $w$ was green. Then $p_2^{-1}$ is obtained by concatenating the path from $v$ to $c$ with $(c, d)$ with the path from $d$ to $b$. The latter consists of $(d, w)$ concatenated with the path from $w$ to $a$ concatenated with the path from $a$ to $b$.

\subsubsection{Finding a maximal set of disjoint paths}
\label{sec.maximal}

After the first path, say $p$, is found, its vertices are removed. As a result, some of the left-over vertices may have no more predecessors. To deal with such vertices, we describe the procedure RECURSIVE REMOVE. It recursively removes all vertices having no more predecessors, until every remaining matched vertex that was assigned a minlevel has a predecessor; of course, unmatched vertices don't have predecessors and will be removed only if they become isolated nodes.  

At this point MAX will process the next bridge of tenacity $l_m$. When it encounters another bridge which makes DDFS terminate in Case 2, i.e., with two unmatched vertices, it finds another augmenting path. This continues until all bridges of tenacity $l_m$ are processed. Lemma \ref{lem.maximal} shows that this will result in a maximal set of paths of length $l_m$.

\section{Limited BFS-Honesty}
\label{sec.BFSH}

Using the notion of tenacity, we first show that minimum length alternating paths are BFS-honest to some extent, see Theorem \ref{thm.honest}. This extent of BFS-honesty will be critically exploited later.

\definition{(Limited BFS-honesty)}
Let $p$ be an $\e(v)$ or $\o(v)$ path starting at unmatched vertex $f$ and let $u$ lie on $p$.
Then $|p[f \ \mbox{to} \ u]|$ will denote the length of this path from $f$ to $u$, and if it is even (odd) we will say
that $u$ is {\em even (odd) w.r.t. p}. 
We will say that $u$ is BFS-honest w.r.t. $p$ if 
$|p[f \ \mbox{to} \ u]| = \e(u) \ (\o(u))$ if $u$ is even (odd) w.r.t. $p$.

Observe that in the graph of Figures \ref{fig.BFSH} and \ref{fig.verten}, the vertices $a, \ b$ and $c$ are BFS-honest on
all evenlevel and oddlevel paths to the vertices of tenacity $\alpha$. However, the vertices of tenacity $\alpha$ are not
BFS-honest on $\o(a)$ and $\o(b)$ paths.

\begin{lemma}
\label{lem.t-matched}
If $(u, v)$ is a matched edge, then $\t(u) = \t(v) = \t(u, v)$.
\end{lemma}

\begin{proof}
If $(u, v)$ is a matched edge, $\e(v) = \o(u) + 1$ and $\e(u) = \o(v) + 1$. The lemma follows.
\end{proof}

As a result of Lemma \ref{lem.t-matched}, in several proofs given in this paper, it will suffice to restrict attention to only 
one of the end points of a matched edge. 

\begin{theorem}
\label{thm.honest}
Let $p$ be an $\e(v)$ or $\o(v)$ path starting at unmatched vertex $f$ and let vertex $u \in p$ with $\t(u) \geq \t(v)$. 
Then $u$ is BFS-honest w.r.t. $p$.
Furthermore, if $\t(u) > \t(v)$ then $|p[f \ \mbox{to} \ u]| = \mn(u)$.
\end{theorem}

\begin{proof}
Assume w.l.o.g. that $p$ is an $\e(v)$ path and that $u$ is even w.r.t. $p$ (by Lemma \ref{lem.t-matched}).
Suppose $u$ is not BFS-honest w.r.t. $p$, and let $q$ be an $\e(u)$ path, i.e., $|q| < |p[f \ \mbox{to} \ u]|$.
First consider the case that $\e(v) = \mx(v)$, and let $r$ be a $\mn(v)$ path. Let $u'$ be the matched neighbor of $u$.
Consider the first vertex of $r$ that lies on $p[u' \ \mbox{to} \ v]$. 
If this vertex is even w.r.t. $p$ then $\o(u) \leq |r| + |p[u \ \mbox{to} \ v]|$. Additionally, 
$\e(u) < |p[f \ \mbox{to} \ u]|$, hence $\t(u) < \t(v)$, leading to a contradiction. On the other hand, if this vertex is odd 
w.r.t. $p$ then $\mn(v) = |r| > \e(u)$, because otherwise there is a shorter even path from $f$ to $v$ than $\e(v)$. 
We combine the remaining argument along with the case that $\e(v) = \mn(v)$ below.

Consider the first vertex, say $w$, of $q$ that lies on $p(u \ \mbox{to} \ v]$ -- there must be such a vertex because otherwise
there is a shorter even path from $f$ to $v$ than $\e(v)$. If $w$ is odd w.r.t. $p$ then we get an even path to $v$
that is shorter than $\e(v)$. Hence $w$ must be even w.r.t. $p$. Then, $q[f \ \mbox{to} \ w] \bullet [w \ \mbox{to} \ u]$ is an odd path to $u$ with length less than $\e(v)$, where $\bullet$ denotes the concatenation operator. Again we get $\t(u) < \t(v)$, leading to a contradiction. 

We next prove the second claim. The claim is obvious if $\e(v) = \mn(v)$, so let us assume that $\e(v) = \mx(v)$. As before,
let $r$ be a $\mn(v)$ path, and consider the first vertex of $r$ that lies on $p[u' \ \mbox{to} \ v]$. 
If this vertex is even w.r.t. $p$ then $\o(u) \leq |r| + |p[u \ \mbox{to} \ v]|$. Hence $\t(u) \leq \t(v)$, which leads to a 
contradiction. On the other hand, if this vertex is odd 
w.r.t. $p$ then $\mn(v) = |r| > \e(u)$, because otherwise there is a shorter even path from $f$ to $v$ than $\e(v)$. Now
the claim follows because otherwise $\t(u) < \t(v)$.
\end{proof}

\begin{corollary}
	\label{cor.Honest}
	Let $p$ be an $\e(v)$ or $\o(v)$ path and let $u$ lie on $p$. If $u$ is not BFS-honest w.r.t. $p$ then $\t(u) < \t(v)$.
\end{corollary}

\section{Base, Blossom and Bridge}
\label{sec.base}

\definition{(Eligible tenacity)} An odd number $t$, with $t_m \leq t < l_m$, will be said to be an {\em eligible tenacity}.  

Let $v$ be vertex of eligible tenacity and $p$ be an $\e(v)$ or $\o(v)$ path; assume it starts at unmatched vertex $f$. If $u$ and $w$ are two vertices on $p$ and if $u$ is further away from $f$ on $p$ than $w$, then we will say that $u$ is {\em higher} than $w$.

\definition{(Base of $v$ w.r.t. $p$, $F(p, v)$)} Let $v$ be vertex of eligible tenacity and $p$ be an $\e(v)$ or $\o(v)$ path starting at unmatched vertex $f$. Consider all vertices of tenacity $> t$ on $p$; clearly, this set contains $f$. Among these vertices, define the highest one to be {\em the base of v w.r.t. p}, denoted $F(p, v)$. Clearly, $F(p, v)$ is even w.r.t. $p$, and by Theorem \ref{thm.honest}, it is an outer vertex.

\notation{(The set $B(v)$)}
For a vertex $v$ of eligible tenacity and let 
\[ B(v) = \{F(p, v) ~|~ p \ \mbox{is an} \ \e(v) \ \mbox{or} \ \o(v) \ \mbox{path} \} . \] 

\bigskip

\begin{claim}
\label{claim.base}
The set $B(v)$ is a singleton.
\end{claim}

The proof of Claim \ref{claim.base} is not straightforward; in fact, the following chicken-and-egg problem arises. The proof of this claim requires the notion of a blossom and its associated properties. On the other hand, blossoms can be defined only after defining the base of a vertex, and the latter definition is a consequence of Claim \ref{claim.base}.

We will break this deadlock by proving Claim \ref{claim.base} via an induction on tenacity: for 
each value of tenacity, say $t$, once Claim \ref{claim.base} is proven for vertices of tenacity $\leq t$, the base of vertices of tenacity $\leq t$ can be defined. Following this, blossoms of 
tenacity $t$ can be defined and properties of these blossoms and properties of paths traversing through these blossoms can be established.
All these facts are then used in the next step of the induction to prove Claim \ref{claim.base} for the next higher value of tenacity.  

\definition{(\Ct)}
Let $t$ be an eligible tenacity. Then \Ct \ is true if Claim \ref{claim.base} holds for all vertices $v$ such that $t_m \leq \t(v) \leq t$.

\begin{figure}[ht]
\begin{minipage}[b]{0.42\linewidth}
\centering
\includegraphics[width=\textwidth]{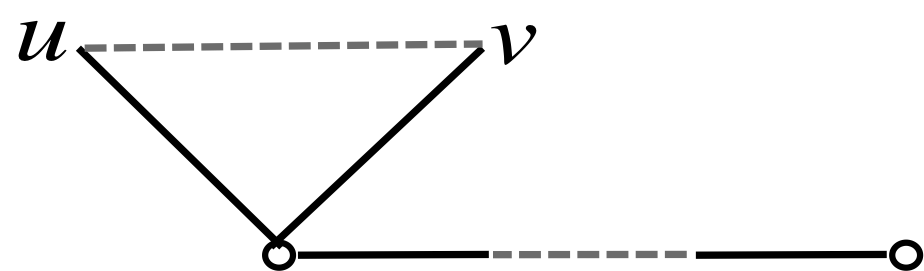}
\caption{
Vertices $u$ and $v$ have no base.}
\label{fig.nobase}
\end{minipage}
\hspace{2.4cm}
\begin{minipage}[b]{0.36\linewidth}
\centering
\includegraphics[width=\textwidth]{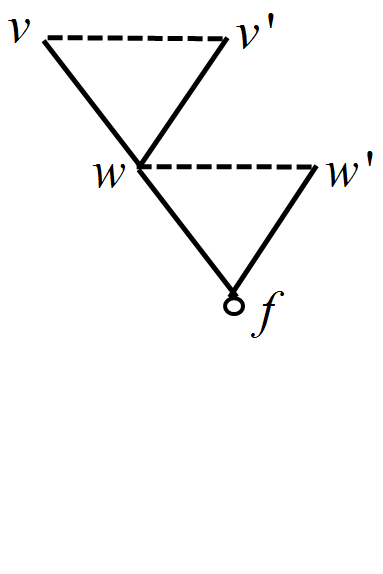}
\caption{Vertex $f$ is the base of $v, v', w, w'$.}
\label{fig.nested2}
\end{minipage}
\end{figure}

\definition{\label{def.base}(Base of a vertex of tenacity $\leq t$)} \\
Let $t$ be an eligible tenacity, and assume \Ct \ holds.  
Then, for each vertex $v$ of tenacity $\leq t$, define its base to be the singleton vertex in the set $B(v)$. We will denote it by $\base(v)$. 

Since $F(p, v)$ is an outer vertex, so is the base of a vertex.
Furthermore, by Theorem \ref{thm.honest}, $\base(v)$ is BFS-honest on every $\e(v)$ or $\o(v)$ path.
In the graph of Figures \ref{fig.BFSH} and \ref{fig.verten}, the base of each vertex of tenacity $\alpha$ is $a$,
tenacity $\kappa$ is $b$, and tenacity $\tau$ is $f$, respectively. In Figure \ref{fig.nested1}, 
$b$ is the base of $v, v', w$ and $w'$. In Figure \ref{fig.nested2}, $f$ is the base of $v, v', w$ and $w'$.

\begin{remark} The condition ``$t < l_m$'' is essential in Claim \ref{claim.base} and in Definition \ref{def.base}, as illustrated in Figure \ref{fig.nobase}. In this graph, $l_m = 3$ and the vertices $u$ and $v$, both of tenacity 3, have no base; the evenlevel and oddlevel paths to these vertices do not contain any vertex of tenacity greater than 3.
\end{remark}

\begin{figure}[ht]
\begin{minipage}[b]{0.28\linewidth}
\centering
\includegraphics[width=\textwidth]{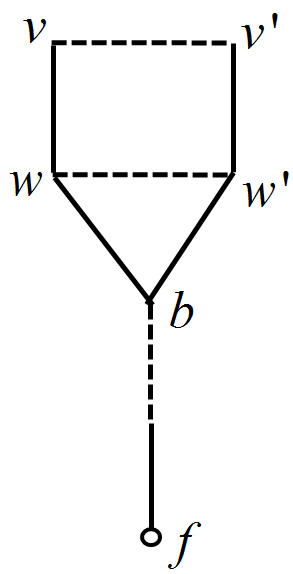}
\caption{Vertex $b$ is the base of $v, v', w, w'$.}
\label{fig.nested1}
\end{minipage}
\hspace{3cm}
\begin{minipage}[b]{0.5\linewidth}
\centering
\includegraphics[width=\textwidth]{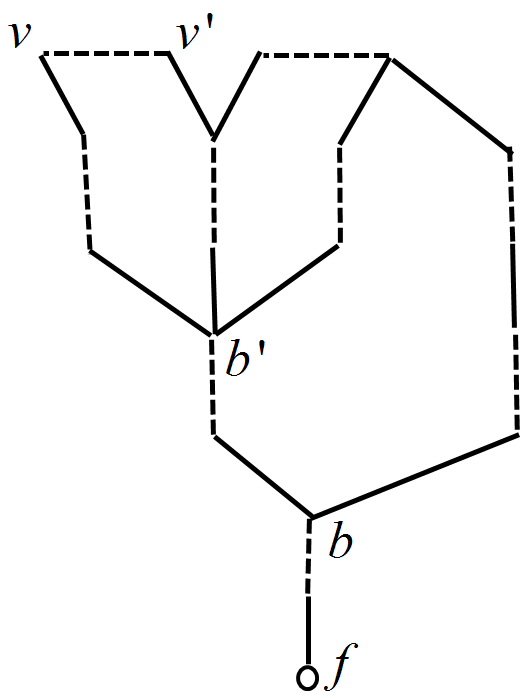}
\caption{Blossom $\CB_{b, 15}$ is of minimum tenacity. Note that $H(v, v') = b'$.}
\label{fig.t-min}
\end{minipage}
\end{figure}

\definition{\label{def.blossom} (Blossom of tenacity $t$)}
Let $t$ be an eligible tenacity, and assume \Ct \ holds. \\
Blossoms  of tenacity $t$ will be defined recursively.
Let $b$ be an outer vertex with $\t(b) > t$. We will denote the {\em blossom of
tenacity $t$ and base $b$} by $\bt$. Define $\CB_{b, 1} = \emptyset$.
Let 
\[S_{b, t} = \{v ~|~ \t(v) = t \ \mbox{and} \ \base(v) = b \} \] 
and define
\[ \bt = S_{b, t} \cup \left(\bigcup_{v \in (S_{b, t} \cup \{b\}), \ v \ \mbox{outer}} \CB_{v, t-2} \right) . \]

It is obvious from Definition \ref{def.blossom} that if $v \in \bt$ then $\t(v) \leq t$. In particular, observe that $b \notin \bt$. 
We will say that blossom $\CB_{b', t'}$ {\em is nested in} blossom $\bt$ if $\CB_{b', t'} \subset \bt$. From the recursive definition
given above, it follows that $t' < t$.

In the graph of Figures \ref{fig.BFSH} and \ref{fig.verten}, the blossom $\CB_{a, \alpha}$ consists of vertices of tenacity $\alpha$,
the blossom $\CB_{b, \kappa}$ consists of vertices of tenacity $\alpha$ and $\kappa$, and
the blossom $\CB_{f, \tau}$ consists of vertices of tenacity $\alpha , \kappa$ and $\tau$. 
In Figure \ref{fig.nested1}, blossom $\CB_{b, 7} = \{w, w'\}$ and blossom $\CB_{b, 11} = \{w, w', v, v'\}$.
In Figure \ref{fig.nested2}, blossom $\CB_{f, 3} = \{w, w'\}$ and blossom $\CB_{f, 7} = \{w, w', v, v'\}$; clearly,
the former is nested in the latter.

\begin{lemma}
\label{lem.free}
Let $v$ be vertex having eligible tenacity. There is an $\e(v)$ path starting at unmatched vertex $f$ if and only if there is an $\o(v)$
path starting at $f$.
\end{lemma}

\begin{proof}
By Lemma \ref{lem.t-matched}, it suffices to prove only the forward direction. By the assumption on $\t(v)$, $v$ has paths of both parity.
Let $p$ be an $\e(v)$ path starting at unmatched vertex $f$ and suppose $q$ is an $\o(v)$ path starting at unmatched vertex $f' \neq f$.
We will show that there must also be an $\o(v)$ path starting at unmatched vertex $f$.

Let $u$ be the lowest vertex w.r.t. $q$ that lies on $p$ as well. If $u$ is even w.r.t. $p$, then there is an augmenting path from $f'$ to $f$ of length
less than $t_m$. Therefore, $u$ is odd w.r.t. $p$. Since $q[f' \ \mbox{to} \ u]  \bullet p[u \ \mbox{to} \ v]$ is a valid even
alternating path,  $|q[f' \ \mbox{to} \ u]|  \geq |p[f \ \mbox{to} \ u]|$, otherwise we can get a shorter $e(v)$ path than $p$.
If $q(u \ \mbox{to} \ v] \cap p[f \ \mbox{to} \ u) = \emptyset$, then $ p[f \ \mbox{to} \ u] \bullet q[u \ \mbox{to} \ v]$ is a valid
odd alternating path of length at most $|q|$ and we are done. 

Otherwise, let $(w' , w)$ be lowest matched edge of $p$ that is traversed by $q$, with $w$ even w.r.t. $p$. If $w$ is odd w.r.t. $q$, then 
again we can get an augmenting path from $f'$ to $f$ of length less than $t_m$. In the remaining case, 
$p[f \ \mbox{to} \ w] \bullet q[w \ \mbox{to} \ v]$ is a shorter odd alternating path than $q$, leading to a contradiction. This proves the lemma.
\end{proof}

Next assume that $v$ has $\e(v)$ and $\o(v)$ paths starting at unmatched vertex $f$. 

\notation{(The set $B_f(v)$)}
Let $f$ be an unmatched vertex and $t$ be an eligible tenacity. Let $v$ be a vertex of tenacity $t$. Define  
\[ B_f(v) = \{F(p, v) ~|~ p \ \mbox{is an} \ \e(v) \ \mbox{or} \ \o(v) \ \mbox{path starting at $f$} \} . \] 

\notation{(The set $V_t(f)$ and the vertex $A_f(v)$)}
Let $f$ be an unmatched vertex and $t$ be an eligible tenacity. We will denote by $V_t(f)$ the set of vertices of tenacity $t$ which have evenlevel and oddlevel paths starting at $f$. For $v \in V_{t}(f)$, consider all $\e(v)$ and $\o(v)$ paths starting at $f$, and consider every vertex of tenacity $> t$ that lies on {\em each} of these paths; in particular, $f$ is such a vertex. Among these vertices, pick the one that is highest and denote it by $A_f(v)$. 

The proof of the next lemma is straightforward and is omitted; the idea behind it is essentially the same as that for Lemma \ref{lem.t-matched}.

\begin{lemma}
\label{lem.min-matched}
Let $(u, v)$ be a matched edge of eligible tenacity. Then $u \in V_{t}(f)$ if and only if $v \in V_{t}(f)$. Furthermore, if $u$ and $v \in V_{t}(f)$, then 
$A_f(u) = A_f(v)$, $B_f(u) = B_f(v)$ and $B(u) = B(v)$. 
\end{lemma}

\subsection{Base Case of the Outer Induction} 
\label{sec.BaseCase-Outer}

In Lemma \ref{lem.min-t}, we establish the base case of the Outer Induction for the first two statements of Theorem \ref{thm.base}. 

\begin{lemma}
\label{lem.min-t}
For every vertex $v$ of tenacity $t_m$, the following hold:
\begin{description}
\item
{\bf Statement 1:}  The set $B(v)$ is a singleton.
\item
{\bf Statement 2:} Every $\mx(v)$ path contains a unique bridge of tenacity $t_m$.
\end{description}
\end{lemma}

\begin{proof}
The proof is quite elaborate and has been split into several claims. 

Let $f$ be an unmatched vertex such that $v \in V_{t_m}(f)$. We will first prove:
\begin{description}
\item
{\bf Statement 1':} The set $B_f(v)$ is a singleton, and $B_f(v) = A_f(v)$. 
\end{description}

\bigskip

\begin{claimb}
\label{claimb.0}
It suffices to prove Statement 1' and Statement 2 for inner vertices in $V_{t_m}(f)$.
\end{claimb}

\begin{proof}
From the proof of Lemma \ref{lem.t-matched}, it is easy to see that the matched neighbor of any outer vertex, say $v$, is an inner 
vertex\footnote{Observe that the reverse is not true, since both endpoints of a matched edge could be inner vertices.}, say $v'$.
Furthermore, by Lemma \ref{lem.min-matched} if $v$ is in $V_{t_m}(f)$, so is $v'$. Next observe that 
every $\mx(v')$ path yields a $\mx(v)$ path by removing the last edge, i.e, $(v, v')$, and 
every $\mx(v)$ path yields a $\mx(v')$ path by concatenating the edge $(v, v')$. Hence proving Statement 1' for $v'$ also proves it for $v$.
For Statement 2, by Lemma \ref{lem.min-matched}, proving it for $v'$ also proves it for $v$.
\end{proof}

Let $v$ be an inner vertex in $V_{t_m}(f)$. We will next define a graph $H_v$ whose structural properties will lead to a proof of the lemma. 
Let $A_f(v) = b$. Consider all $\e(v)$ and $\o(v)$ paths starting at $f$, and denote by $H_v$ all the vertices and edges on these paths.
Let $\B = \mn(b)$ and $\A = (t_m - 1)/2$; clearly $\A$ is the maximum possible minlevel of a vertex of tenacity $t_m$. 
By definition of $t_m$, every vertex on an $\e(v)$ or $\o(v)$ path is BFS-honest on it. 
$H_v$ is a layered graph, similar to graph $H$ defined in Section \ref{sec.DDFS}, with one exception (made for the sake of
convenient notation), namely it will have edges running between pairs of vertices at layer $\A$. $H_v$ has $\A + 1$ layers numbered $\A$ to 0. 

The layer of each vertex in $H_v$ is defined to be its minlevel. Thus layer 0 has only $f$ in it. Clearly, all the props run between adjacent 
layers of $H_v$. Furthermore, $H_v$ satisfies the DDFS Requirement (stated in Section \ref{sec.DDFS}), 
that starting from any vertex, there is a path to the lowest layer, i.e., vertex $f$.

\begin{claimb}
\label{claimb.bridge}
Let $p$ be an $\e(v)$ path in $H_v$. Then $p$ contains a bridge of tenacity $t$. 
\end{claimb}

\begin{proof}
Clearly, $p$ will contain an edge $(u, u')$ such that $\mn(u) = \mn(u') = \A$. We will show that $(u, u')$ is a bridge of tenacity $t$. 
Any $\mn(u)$ path concatenated with the edge $(u, u')$ gives a $\mx(u')$ path and vice versa.
Hence $\t(u) = \t(u') = t_m$. Furthermore, the predecessors of $u$ and $u'$ are vertices at with minlevel $\A - 1$. Hence $(u, u')$ is a bridge. 
\end{proof}

Let $(u, u')$ be a bridge of tenacity $t_m$ in $H_v$. Consider all possible $\mn(u)$ and $\mn(u')$ paths and say that vertex $w$ is 
a {\em bottleneck} if it lies on all such paths. Denote the highest (in minlevel) bottleneck by $H(u, u')$; clearly, this will be an outer vertex. 

\begin{claimb}
\label{claimb.1}
$H(u, u')$ lies on every $\o(v)$ path.
\end{claimb}

\begin{proof}
Suppose not, and consider an $\o(v)$ path that does not use $w$. This path can be extended to a
$\mn(u)$ or $\mn(u')$ path, thereby contradicting the fact that $H(u, u')$ is a bottleneck for all such paths. 
\end{proof}

Define the set $S(u, u')$ as follows. For each $\mn(u)$ and $\mn(u')$ path, include in $S(u, u')$ all vertices that are higher than $H(u, u')$. 

\begin{claimb}
\label{claimb.S}
For each vertex $w \in S(u, u')$, $\t(w) = t_m$. 
\end{claimb}

\begin{proof}
Clearly, DDFS run on $H_v$ with bridge $(u, u')$ must end with the bottleneck $H(u, u')$, and it will visit each vertex in $S(u, u')$. 
Therefore, for each $w \in S(u, u')$, the concatenation of the two paths established by the DDFS Certificate, together an $\e(H(u, u'))$ path
and the edge $(u, u')$, yields a $\mx(w)$ path, which shows that $\t(w) = t_m$. 
\end{proof}

\medskip

\subsubsection{The Inner Induction}
\label{sec.Inner-Induction}

A proof of Statement 1' requires an induction on $\mn(v)$, for $v \in V_{t_m}(f)$. We will show that for each $l \in [\B + 1, \A]$, all vertices in $H_v$ that have minlevel $l$ must have tenacity $t_m$, thereby proving that $B_f(v) = \{b\}$. 

{\bf Basis for the Inner Induction:} 
Let $v$ be a vertex of minimum minlevel in $V_{t_m}(f)$; clearly $v$ is inner. Let $A_f(v) = b$,
defined above. 

\begin{claimb}
\label{claimb.2}
Let $p$ be an $\o(v)$ path. Then the last edge of $p$ is $(b, v)$. Furthermore, there is a bridge $(u, u')$ in $H_v$ such that $H(u, u') = b$.
\end{claimb}

\begin{proof}
Suppose not, and let the last edge of $p$ be $(w, v)$. Now, an $\e(v)$ path concatenated with edge $(v, w)$ gives an odd
alternating path to $w$ thereby proving that $\t(w) = t_m$. However, $\mn(w) < \mn(v)$, which contradicts the choice of $v$,
thereby proving that the last edge of $p$ is $(b, v)$.

Let $(u, u')$ be a bridge on any $\e(v)$ path, say $q$. Since $p[b \ \mbox{to} \ v]$ is simply the edge $(b, v)$, the only bottleneck on
$q[b \ \mbox{to} \ v]$ is $b$. Hence $H(u, u') = b$.
\end{proof}

By Claim \ref{claimb.2}, for any $\e(v)$ path $p$, every vertex on $p(b \ \mbox{to} \ v]$ has tenacity $t_m$. 
Therefore $B_f(v) = A_f(v) = \{b\}$, hence proving the induction basis.

\medskip

{\bf Step for the Inner Induction:} 
Let $v$ be an inner vertex in $V_{t_m}(f)$ with $\mn(v) = l$, where $l \in (\B + 1, \A]$, and assume that Statement 2 holds for every 
vertex in $V_{t_m}(f)$ having minlevel $< l$. Consider all bridges of tenacity $t_m$ that lie on $\e(v)$ paths. For each one,
say $(u, u')$, consider the bottleneck $H(u, u')$. Among these bottlenecks, let $w$ be one having lowest minlevel.  
As shown in Claim \ref{claimb.1}, $w$ lies on all $\o(v)$ paths. 

\begin{claimb}
\label{claimb.3}
$w$ lies on every $\e(v)$ path.
\end{claimb}

\begin{proof}
Suppose $p$ is an $\e(v)$ path that does not contain $w$, and let $(u, u')$ be the bridge on this path. Now there are two cases: either there 
is an $\o(v)$ path $q$ such that $p$ shares a matched edge $(y, y')$ on $q(w \ \mbox{to} \ v)$ or there is no such $\o(v)$ path. In the first case,
$p[f \ \mbox{to} \ y] \bullet q[y \ \mbox{to} \ v]$ is an $\o(v)$ path not using $w$. In the second case, $H(u, u')$ is below $w$.
Both cases lead to contradictions. 
\end{proof}

Since $w$ is the highest bottleneck for all $\o(v)$ and $\e(v)$ paths, any vertex $z \in G_v$ with $\mn(z) > \mn(w)$, $z \in S(u, u')$ for 
some bridge $(u, u')$ of tenacity $t_m$ that lies on an $\e(v)$ path. Therefore, by Claim \ref{claimb.S}, $\t(z) = t_m$. 
Now there are two cases: $w = b$ and $w \neq b$. In the first case, we have established that $B_f(v) = A_f(v) = b$.

\begin{claimb}
\label{claimb.4}
If $w \neq b$, then $\t(w) = t_m$ and $B_f(w) = A_f(w) = b$.
\end{claimb}

\begin{proof}
If $\t(w) > t_m$, $A_f(v) = w$, leading to a contradiction. 
Furthermore, since $mn(w) < \mn(v)$, the induction hypothesis applies to $w$, and $B_f(w) = A_f(w)$. 

We next establish that $A_f(w) = b$.
Every $\e(w)$ path, $p$, can be extended to an $\o(v)$ path and therefore $F(p, w) = b$.
Suppose there is an $\o(w)$ path $q$ with $F(q, w) \neq b$. Then either there is an $\e(w)$ path $p$
such that $q$ shares a matched edge $(y, y')$ on $p(b \ \mbox{to} \ w)$ or there is no such $\e(w)$ path. In the first case,
$r = q[f \ \mbox{to} \ y] \bullet p[y \ \mbox{to} \ w]$ is an $\e(w)$ path such that $F(r, w) \neq b$. In the second case, 
$q \bullet p[w \ \mbox{to} \ b]$ is an $\o(b)$ path showing that $\t(b) = t_m$; here $p$ is any $\e(w)$ path.
Both cases lead to contradictions. Therefore $F(q, w) = b$, and hence $A_f(w) = b$. 
\end{proof}

\begin{claimb}
\label{claimb.5}
If $w \neq b$, then for every $\e(v)$ and $\o(v)$ path $p$, every vertex on $p(b \ \mbox{to} \ w)$ is of tenacity $t_m$.
\end{claimb}

\begin{proof}
Since $w$ occurs on every $\e(v)$ and $\o(v)$ path $p$ and is BFS-honest on it, $p$ consists of
an $\e(w)$ path concatenated with a appropriate path from $w$ to $v$.
Since $B_f(w) = b$, for any $\e(w)$ path $q$, every vertex on $q(b \ \mbox{to} \ w)$ is of tenacity $t_m$.
Hence every vertex on $p(b \ \mbox{to} \ w)$ is of tenacity $t_m$.
\end{proof}

Therefore we have established that $B_f(v) = A_f(v) = \{b\}$  in second case as well, i.e., $w \neq b$.
This completes the proof of Statement 1'. The next claim will enable us to prove Statement 1.

\begin{claimb}
\label{claimb.6}
Let $f$ and $f'$ be unmatched vertices. For every vertex $v \in V_{t_m}(f) \cap V_{t_m}(f')$, the following holds:
$B_f(v) = B_{f'}(v)$.
\end{claimb}

\begin{proof}
Suppose $B_f(v) \neq B_{f'}(v)$, and let $b = B_f(v)$ and $b' = B_{f'}(v)$, with $\mn(b) \leq \mn(b')$.
All $\e(v)$ and $\o(v)$ paths from $f$ use $b$ and not $b'$, and those from $f'$ use $b'$ and not $b$.
In particular, any $\e(b)$ path is vertex disjoint from any $\e(b')$ path. 

Let $S_1 \subseteq V_{t_m}(f)$ be the set of vertices\footnote{As shown in Figure \ref{fig.subset}, $S_1 \subset V_{t_m}(f)$
is possible.}
whose evenlevel and oddlevel paths contain $b$, and
$S_2 \subseteq V_{t_m}(f')$ be the set of vertices whose evenlevel and oddlevel paths contain $b'$. 
Let $S$ be the set of vertices of minimum minlevel in $S_1 \cap S_2$; clearly they are all inner. First consider the case that 
there is $v \in S$ that is adjacent to $b$ or $b'$, say the former. Then an $\o(v)$ path containing $b$ concatenated with an $\e(v)$ path 
containing $b'$ will be a simple alternating path and hence an augmenting path from $f$ to $f'$ of length $t_m$, contradicting the 
assumption $t_m < l_m$. 

Pick any $v \in S$ and let $(w, v)$ be the last edge on an $\o(v)$ path that contains $b$. Let $p$ be an $\e(v)$ 
path that contains $b'$. Then $p$ concatenated with the edge $(v, w)$ yields an $\o(w)$ path that contains $b'$.
Now applying Lemma \ref{lem.free} we get that $w \in S_1 \cap S_2$. Since $\mn(w) < \mn(v)$, we get a contradiction.
Hence $B_f(v) = B_{f'}(v)$.
\end{proof}

Claim \ref{claimb.6} immediately implies that $B(v)$ is a singleton, thereby proving Statement 1. 
It further implies that $S_1 = S_2$ and hence Statement 2 follows from Claim \ref{claimb.bridge}. 
\end{proof}

\begin{remark}  It is easy to construct an example in which $u, v \in V_{t_m} (f)$ and $B_f(u) \neq B_f(v)$.
\end{remark}


\begin{figure}[h]
\begin{center}
\includegraphics[scale = 0.4]{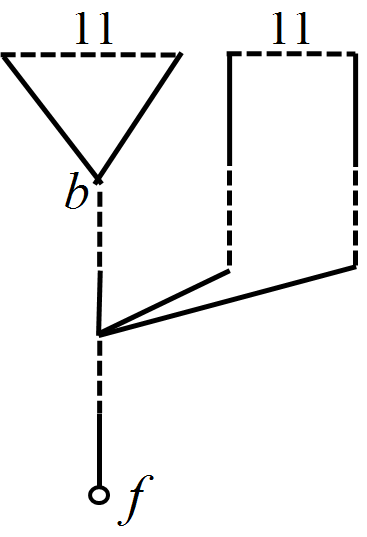}
\caption{$\CB_{b,11}$ is a proper subset of $V_{11}(f)$.}
\label{fig.subset}
\end{center}
\end{figure}


\subsection{Iterated Bases of a Vertex}
\label{sec.iterated-bases}

\definition{(Iterated bases of a vertex of tenacity $\leq t$)}\\
\label{def.iterated}
Let $t$ be an eligible tenacity, and assume \Ct \ holds. \\
Let $v$ be a vertex such that $\t(v) = t' \leq t$. The following are the iterated bases of $v$:\\
Define $\base^1(v) = \base(v)$, and for $k \geq 1$, if $\t(\base^k(v)) \leq t$, then define $\base^{k+1}(v) = \base(\base^k(v))$.

In the graph of Figures \ref{fig.BFSH} and \ref{fig.verten}, $\base^1(v) = a, \base^2(v) = b$ and $\base^3(v) = f$.

\definition{(Shortest path from an iterated base to a vertex)} \\
Let $t$ be an eligible tenacity, and assume \Ct \ holds. \\
Let $v$ be a vertex such that $\t(v) \leq t$, and let $k \in \Zplus$ such that $\t(\base^k(v)) \leq t$. Let $\base^{k+1}(v) = b$.
Then by an $\e(b; v)$ ($\o(b; v)$) path we mean a minimum even (odd) length alternating path from $b$ to $v$ that starts with an 
unmatched edge. 

\begin{proposition}
\label{prop.concat}
Let $t$ be an eligible tenacity, and assume \Ct \ holds. \\
Let $v$ be a vertex such that $\t(v) = t$ and $\base(v) = b$. Then every $\e(v)$ ($\o(v)$) path consists of an $\e(b)$ path 
concatenated with an $\e(b; v)$ ($\o(b; v)$) path.
\end{proposition}

\begin{proof}
Let $p$ be an $\e(b)$ path starting at unmatched vertex $f$ and $q$ be an $\e(b; v)$ path. If their concatenation is longer than 
$\e(v)$ then $q$ must intersect $p$ below $b$. Let $(w, w')$ be the lowest matched edge of $p$ used by $q$, where $w'$ is even w.r.t. $p$. 
Now, using the same arguments as those in the proof of Theorem \ref{thm.honest},
one can show that if $w$ is odd w.r.t. $q$ then there is an even path from $f$ to $v$ that is shorter than $\e(v)$ and
if $w'$ is odd w.r.t. $q$ then there is a short enough odd path from $f$ to $b$ which gives $\t(b) \leq t$. 
\end{proof}

The next lemma can be proven unconditionally and will be needed critically in Theorem \ref{thm.base}.

\begin{lemma}
\label{lem.stretch}
Let $v$ be vertex of eligible tenacity $t$, $p$ be a $\mn(v)$ path, and $b = F(p, v)$. Let $u$ be a vertex on $p[b \ \mbox{to} \ v]$ satisfying $\t(u) < t$, and assume that edge $(u, w)$ lies on $p$ with $w$ higher than $u$ and $\t(w) = t$. Then $u$ is BFS-honest w.r.t. $p$.
\end{lemma}

\begin{proof}
Assume w.l.o.g. that $p$ is an $\e(v)$ path (by Lemma \ref{lem.t-matched}). Clearly, $u$ is even w.r.t. $p$ and $w$ is odd w.r.t. $p$.
Suppose $u$ is not BFS-honest w.r.t. $p$, and let $q$ be an $\e(u)$ path, i.e., $|q| < |p[f \ \mbox{to} \ u]|$.

Consider the first vertex, say $z$, of $q$ that lies on $p(u \ \mbox{to} \ v]$ -- there must be such a vertex because otherwise
there is a shorter even path from $f$ to $v$ than $\e(v)$. If $z$ is odd w.r.t. $p$ then again we get an even path to $v$
that is shorter than $\e(v)$. Hence $z$ must be even w.r.t. $p$. Then, $q[f \ \mbox{to} \ z] \bullet p[z \ \mbox{to} \ w]$
is an even path to $w$ with length less than $\mn(v)$. Furthermore, since $w$ is odd w.r.t. $p$, $\o(w) < \mn(v)$.
Hence $\t(w) < \t(v)$, leading to a contradiction. 
\end{proof}


\subsection{Completing the Outer Induction}
\label{sec.Outer-Ind}

\begin{theorem}
\label{thm.base}
Let $t$ be an eligible tenacity, and let $v$ be a vertex of tenacity $t$. The following hold:
\begin{description}
\item
{\bf Statement 1:} The set $B(v)$ is a singleton.
\item
{\bf Statement 2:} Every $\mx(v)$ path contains a unique bridge of tenacity $t$.
\item
{\bf Statement 3:} The blossoms of tenacity at most $t$ form a laminar family.
\item
{\bf Statement 4:} Let $\base(v) = b$ and $u \in \bt$. Then every $\e(u)$ ($\o(u)$) path consists of an $\e(b)$ path
concatenated with an $\e(b; u)$ ($\o(b; u)$) path. Moreover, every $\e(b; u)$ and $\o(b; u)$ path lies in $\bt \cup \{b\}$.
Furthermore, every edge on every $\e(b; u)$ and $\o(b; u)$ path is of tenacity $\leq t$.
\item
{\bf Statement 5:} Let $\base(v) = b$, $p$ be an $\e(v)$ or $\o(v)$ path, and $u$ lie on $p(b \ \mbox{to} \ v]$.
Then $\base(u)$ lies on $p[b \ \mbox{to} \ v]$.
\end{description}

Additionally, for every vertex $v$ of tenacity $l_m$, every $\mx(v)$ path contains a unique bridge of tenacity $l_m$.
\end{theorem}

\begin{proof}  \\
{\bf Basis for the Outer Induction:} Statements 1 and 2 are proven in Lemma \ref{lem.min-t}. As a result $\base(v) = b$, say, and blossom 
$\CB_{b, t_m}$ are unconditionally defined. If $d$ and $d'$ are distinct vertices of tenacity $> t_m$ then the blossoms $\CB_{d, t_m}$ and
$\CB_{d', t_m}$ are disjoint sets, since a vertex cannot have both $d$ and $d'$ as its base, leading to a proof of Statement 3 holds. 
Observe that vertex $u$ in Statement 4 must have tenacity $t_m$. 
The first part of this statement follows from Proposition \ref{prop.concat}. For the second part, by Lemma \ref{lem.min-t}, 
every vertex $w \neq b$ on an $\e(b; u)$ or $\o(b; u)$ path has tenacity $t_m$ and base $b$, i.e., lies in $\CB_{b, t_m}$. 
The third part is easy to see from the structural properties of graph $H_v$ established in Lemma \ref{lem.min-t}.
Hence Statement 4 holds. Vertex $u$ in Statement 5 must also have tenacity $t_m$, and hence its proof follows from Statement 4.

\bigskip

{\bf Step for the Outer Induction:} Let $t$ be an eligible tenacity, and assume that the theorem holds for tenacity $< t$.

\begin{claimc}
\label{claimc.1}
{\bf Statement 2}  holds.
\end{claimc}

\begin{proof}
Let $p$ be a $\mx(v) = \e(v)$ path and $q$ be a $\mn(v)$ path; clearly $|p| + |q| = t$.
By Theorem \ref{thm.honest}, each vertex $u$ of tenacity $\geq t$ on $p$ is BFS-honest w.r.t. $p$. Among these
let us partition the vertices having minlevel $\geq \B$ into two sets: $T_1$ ($T_2$) consists of vertices $u$ such that 
$|p[f \ \mbox{to} \ u]| = \mn(u)$ ($ = \mx(u)$). Clearly $b \in T_1$ and $v \in T_2$, hence both sets are non-empty. 
Let $a$ be the vertex in $T_1$ having the largest minlevel and $c$ be the vertex in $T_2$ having the smallest maxlevel. 
Now there are three cases.

{\bf Case 1:} $a$ and $c$ are adjacent on $p$ and $(a, c)$ is a matched edge. Since $a \in T_1$ and $c \in T_2$, $\t(a) \geq t$
and $\t(c) \geq t$. The concatenation of $p$ and $q$ can be viewed as the concatenation of an even and an odd path from $f$ to $a$, 
giving $\t(a) \leq t$. Combined with the previous statement we get $\t(a) = t$. By Lemma \ref{lem.t-matched} we get
$\t(a) = \t(c) = \t(a, c) = t$. Since $p$ assigns minlevel to $a$ and maxlevel to $c$, it must be the case that $\mn(a) = \A$ 
and $\mx(c) = \A + 1$. Furthermore, since $\t(a) = \t(c) = t$, we get $\mn(c) = \A$ and $\mx(a) = \A + 1$, i.e.,
$a$ and $c$ are both inner. Therefore $(a, c)$ is not a prop, and hence it is a bridge of tenacity $t$. 

{\bf Case 2:} $a$ and $c$ are adjacent on $p$ and $(a, c)$ is an unmatched edge. By arguments analogous to the previous case,
it is easy to see that $\t(a) = \t(c) = \t(a, c) = t$, and that $a$ and $c$ are both outer.
Therefore $(a, c)$ is not a prop, and hence it is a bridge of tenacity $t$. 

{\bf Case 3:} In the remaining case, $a$ and $c$ are not adjacent on $p$, i.e., there are vertices of tenacity $< t$ between $a$ and 
$c$ on $p$. Let $(c, c')$ be the unmatched edge on $p$. There are two cases:

{\bf Case 3a:} $(c, c')$ is a prop. If so, $c'$ lies in the blossom $\CB_{c, t-2}$. Let $(d, d')$ be the first edge on $p[c \ \mbox{to} \ f]$
that ``comes out of this blossom,'' i.e., $d \in \CB_{c, t-2}$ and $d' \notin \CB_{c, t-2}$. Clearly $(d, d')$ is unmatched.
Now there are two cases. If $d' = a$, then $(d, a)$ must be a bridge. The reason is that $a$, which gets its minlevel from $p$
must be outer, and all predecessors of $d$ lie inside $\CB_{c, t-2} \cup \{c\}$. 

Now, the concatenation of $p$ and 
$q$ can be viewed as the concatenation of an $\e(d)$ path, an $\e(a)$ path and the edge $(d, a)$ -- the only aspect requiring
justification is the first path, which follows from Statement 5 of the induction hypothesis applied to vertex $d \in \CB_{c, t-2}$.
Hence $(a, d)$ is a bridge of tenacity $t$. 

If $d' \neq a$, $\t(d') < t$ and so $d'$ lies in a blossom of tenacity $t-2$, say $\CB_{e, t-2}$, having base $e$. Now, by Statement 5 of
the induction hypothesis applied to $d' \in  \CB_{e, t-2}$ we get that every shortest even path from $d'$ to $f$
must use $e$. Hence the $e$ lies on $p[d' \ \mbox{to} \ a]$. Furthermore, $\t(e) \geq t$. Therefore $e = a$. 
Finally, by analogous arguments to the previous case, the concatenation of $p$ and 
$q$ can be viewed as the concatenation of an $\e(d)$ path, an $\e(d')$ path and the edge $(d, d')$, which shows that
$(d, d')$ is a bridge of tenacity $t$. 

{\bf Case 3b:} $(c, c')$ is a bridge. By similar arguments to the previous case we get that $c' \in \CB_{a, t-2}$ and that
$(c, c')$ is a bridge of tenacity $t$. 

Finally, we show that none of the remaining edges on $p$ is a bridge of tenacity $t$. Consider an edge $(d, e)$ on $p[f \ \mbox{to} \ a]$,
with $d$ below $e$ on $p$. If $\t(e) \geq t$ then $e$ is BFS-honest on $p$ and $(d, e)$ is a prop. If $t(d) < t$ then $d$ lies in a blossom 
of tenacity $t-2$ and hence by Statement 5 of the induction hypothesis, $\t(d, e) < t$. A similar argument holds for the edges 
on $p[c \ \mbox{to} \ v]$. This completes the proof of Statement 2.
\end{proof}

Let $f$ be an unmatched vertex such that $v \in V_{t}(f)$.
The structure of the proof of the next Claim is along the lines of the proof of Statement 1' in Lemma \ref{lem.min-t}.
Therefore, in the proof given below, our emphasis is on providing only the new ideas needed.

\begin{claimc}
\label{claimc.2}
The following holds:
\begin{description}
\item
{\bf Statement 1':} The set $B_f(v)$ is a singleton, and $B_f(v) = A_f(v)$. 
\end{description}
\end{claimc}

\begin{proof}
Once again it suffices to consider inner vertices only, see Claim \ref{claimb.0}. 
Let $v$ be an inner vertex in $V_{t}(f)$. We will define a graph $H_v$ whose structural properties will lead to a proof of 
Statement 1'. Let $A_f(v) = b$. Let $\B = \mn(b)$ and $\A = (t_m - 1)/2$; clearly $\A$ is the maximum possible minlevel of a vertex of 
tenacity $t$. $H_v$ is a layered graph, similar to graph $H$ defined in Section \ref{sec.DDFS}.

The edges of $H_v$ go from higher layers to lower layers and are not distinguished as matched or unmatched. Graph $H_v$ has 
$\A + 1$ layers, which are numbered from 0 to $\A$, with not necessarily all layers having vertices. 
As in Lemma \ref{lem.min-t}, we will make one exception to edges not running within the same layer:  the two end points 
of certain bridges of tenacity $t$ may both lie in layer $\A$. If so, we will add an edge connecting them.

The manner in which $H_v$ is obtained from the original graph $G = (V, E)$ is specified below. Its vertices will be a suitably chosen subset of $V$.  
The layer number of each vertex $w \in G_v$ is $\mn(w)$; in particular, layer 0 will contain only $f$. Each edge also has a specified length; for
edges not having layer $\A$ as one of the end points, the length of the edge is the difference of the layer numbers of its end points. 

We will show that $H_v$ satisfies the DDFS Requirement (stated in Section \ref{sec.DDFS}), namely starting from any vertex, there is a path to 
the lowest layer. Additionally, we will prove the following correspondence between paths in $H_v$ and alternating paths in $G = (V, E)$.

{\bf Correspondence of paths between $H_v$ and $G$:}
Corresponding to each simple path in $H_v$ from $u$ in layer $l$ to $w$ in layer $l'$, with $l > l'$, such that each edge of the path
goes from a higher to a lower layer, there is a simple alternating path of the same length in $G = (V, E)$.  
Furthermore, if the DDFS Guarantee gives disjoint paths from $a$ to $u$ and $c$ to $w$, for some bridge $(a, c)$ of tenacity $t$,
then there are vertex-disjoint simple alternating paths from $a$ to $u$ and $c$ to $w$ in $G = (V, E)$ 
of the same lengths.

Consider all $\e(v)$ and $\o(v)$ paths starting at $f$. By Theorem \ref{thm.honest}, every vertex of tenacity $\geq t$ on such a path 
is BFS-honest on it. Let $S$ denote all such vertices.  The vertex set of $H_v$ is $S \cup S'$, where $S'$ will be defined below. Its
edge set is $E' \cup E''$, where $E'$ is a specially chosen subset of $E$, and $E''$ are additional edges defined below. 
The length of each edge in $E'$ is unit and for edges in $E''$, the length is specified below. 
For each pair of vertices $u, w \in S$, if $(u, v) \in E$ and $\mn(u) \neq \mn(v)$, then $(u, v)$ is included in $E'$. In addition $E'$ will
contain all bridges of tenacity $t$, as pointed out below.

Next we define the edges of $E''$. Intuitively, these edges will replace ``sub-paths that lie inside blossoms of tenacity $t-2$.'' 
Let $w \in S$ and let $p$ be a $\mn(w)$ path that starts at $f$; clearly $p$ is part of an $\e(v)$ 
or $\o(v)$ path. Let $a$ and $c$ be vertices of tenacity $< t$ on $p$, with $a$ lower than $c$. Let $a'$ immediately precede $a$ and $c'$
immediately succeed $c$ on $p$. We will say that $p[a \ \mbox{to} \ c]$ is a {\em maximal contiguous stretch of vertices of tenacity 
$< t$ on $p$} if $\t(a') \geq t$, $\t(c') \geq t$ and all vertices on $p[a \ \mbox{to} \ c]$ are of tenacity $< t$.
By Lemma \ref{lem.stretch}, $c$ is BFS-honest w.r.t. $p$ and by Lemma \ref{lem.t-matched}, $|p[f \ \mbox{to} \ c]|$ must be even.
Hence by Statement 4 of the induction hypothesis, $p[a' \ \mbox{to} \ c] = \e(a'; c)$ and it lies in $\CB_{a', t-2} \cup \{a'\}$. 

Now, in lieu of the path $p[a' \ \mbox{to} \ c']$, the direct edge $(a', c')$ is added to $E''$, and its length is defined to be $|p[a' \ \mbox{to} \ c']|$.
Observe that the length equals the difference in the layer numbers of $c'$ and $a'$. 
Thus edge $(a', c')$ of graph $H_v$ 
represents the path $\e(a'; c) \bullet  (c, c')$ in the original graph $G = (V, E)$. This operation is performed on
every relevant sub-path of every $\e(v)$ and $\o(v)$ path. 

Finally we add vertices and edges to $H_v$ corresponding to each bridge of tenacity $t$; the precise addition depends on the
case in Claim \ref{claimc.1} satisfied by this bridge. In Case 1 and 2, we only need to add the edge $(a, c)$.

In the first case within Case 3a, we add vertex $d$ to $S'$ and assign it layer $\A$.
We also add edge $(c, d)$ to $E''$ and $(d, a)$ to $E'$. The length of $(c, d)$ is $|p[c \ \mbox{to} \ d]| = \e(c; d)$ and it
corresponds to an $\e(c; d)$ path in $G = (V, E)$.

In the second case within Case 3a, we add $d$ and $d'$ to $S'$, both at layer $\A$ and we add $(d, d')$ to $E'$. 
We also add edges $(c, d)$ and $(d' a)$ to $E''$. These correspond to an $\e(c; d)$ path and an $\e(a; d')$ path in $G = (V, E)$, respectively,
and their lengths are defined to be the lengths of these paths. 

In case 3b, we add $c'$ to $S'$, at layer $\A$, together with edge $(c, c')$ in $E'$. We also add edge $(c', a)$ to $E''$; it corresponds to 
an $\e(a; c')$ path in $G = (V, E)$ and its length is defined to be the length of this path. 

This completes the description of graph $H_v$. It is easy to verify that $H_v$ satisfies the DDFS Requirement, stated in Section \ref{sec.DDFS}). We now prove that the
correspondence of paths between $H_v$ and $G = (V, E)$ holds. Observe that if $(a, c)$ and $(d, e)$ are edges
in $E''$ on four distinct vertices, then they correspond to two paths that lie in two distinct blossoms of tenacity $t-2$. 
By Statement 3 of the induction hypothesis, i.e., laminarity  of blossoms, these two blossoms are vertex disjoint, thereby implying that 
the required paths through them are also vertex disjoint. 

We note that the remaining ideas needed to complete the proof of Statement 1' are identical to those used for proving 
Statement 1' in Lemma \ref{lem.min-t}.
\end{proof}

Again, Claim \ref{claimb.6} extends Statement 1' to Statement 1. 
At this point, the base of every vertex of tenacity $t$, and blossoms of tenacity $t$ are unconditionally defined. Hence
Statement 3 follows from Proposition \ref{prop.laminar}. 

The first part of Statement 4 follows from Proposition \ref{prop.concat}. The second part follows from the fact that all vertices
of tenacity $< t$ on $\e(v)$ and $\o(v)$ paths lie in blossoms of tenacity $t-2$, which by the definition of blossoms will be
nested inside $\bt$. The additional vertices referred to in the third part are those of tenacity $< t$ in $\bt$. Such a 
vertex $u$ lies in a blossom $\CB_{d, t-2}$, where $d$ is either $b$ or $d$ is a vertex of tenacity $t$ in $\bt$. The first case,
follows by Statement 4 of the induction hypothesis, and in the second case,
$b$ is an iterated base of $u$ and an $\e(b; d)$ path concatenated with an appropriate
path in $\CB_{d, t-2}$, which is guaranteed by Statement 4 of the induction hypothesis, yields the required path.
The structure of $H_v$ readily implies the third part of Statement 4, hence proving this statement fully.

If $\t(u) = t$, Statement 5 is obvious. Next assume $\t(u) < t$. Let $w$ be the first vertex on $p[u \ \mbox{to} \ v]$ having tenacity
$t$ and let $w'$ be the preceding vertex on this path; clearly $\t(w') < t$. Also, let $a$ be the last vertex on $p[b \ \mbox{to} \ u]$
of tenacity $\geq t$. By Lemma \ref{lem.stretch}, $w'$ is BFS-honest on $p$ and by Statement 4 of the induction hypothesis, 
$p[a \ \mbox{to} \ w']$ is an $e(a; w')$ path. Now, by Statement 5 of the induction hypothesis, $\base(u)$ lies on $p[a \ \mbox{to} \ w']$,
thereby completing the proof of Statement 5. This also completes the proof of the induction step.

To establish the last claim made in the theorem, let $v$ be a vertex of tenacity $l_m$. Observe that the arguments made in Claim \ref{claimc.1} for 
proving Statement 2 do not hinge on proving any of the subsequent statements for that value of tenacity. Hence the proof of this claim will
work for vertices of tenacity $l_m$ as well.
\end{proof}


\suppress{

Theorems \ref{thm.honest} and \ref{thm.base} give:

\begin{corollary}
\label{cor.bases}
For any vertex $v$, its iterated bases occur on every $\e(v)$ and $\o(v)$ path in a BFS-honest manner.
\end{corollary}

}

As stated in Section \ref{sec.Petal-Bud}, the notions of petal and bud are intimately related to the notions of blossom and base. This relationship is formally established in Lemma \ref{lem.all}. At the end of search level $i = (t-1)/2$, i.e., once MAX is done processing all bridges of tenacity $t$, all blossoms of tenacity $t$ can be identified as follows. The proof of this lemma is straightforward and is omitted.

\begin{lemma}
\label{lem.all}
Let $\t(v) = t$, and at the end of search level $i = (t-1)/2$, assume that $\bds(v)$ is $b$. 
Then $\base(v) = b$ and the set $S_{b, t}$ defined in 
Definition \ref{def.blossom}, for blossom $\bt$ is precisely $\{ u ~|~ \t(u) = t \ \mbox{and} \ \bds(u) = b \}$.
Furthermore, blossom $\bt$ consists of the union of all petals whose \bds \ is $b$ at the end of search level $i = (t-1)/2$, together
with each blossom of tenacity $(t - 2)$ whose base is $b$ or any of the vertices of these petals.
\end{lemma}

Observe that if $\bds(v)$ is computed at the end of search level $j > i$, then it may not be $b$ anymore. However, it will be an
iterated base of $v$.


\section{Blossoms Form a Laminar Family}
\label{sec.laminar}

In this section, we will assume that \Ct \ holds for eligible tenacity $t$, and we will show that 
the set of blossoms of tenacity at most $t$ forms a laminar family. This will prove Statement 3 in the induction step in Theorem \ref{thm.base}.

\begin{figure}[ht]
\begin{minipage}[b]{0.28\linewidth}
\centering
\includegraphics[width=\textwidth]{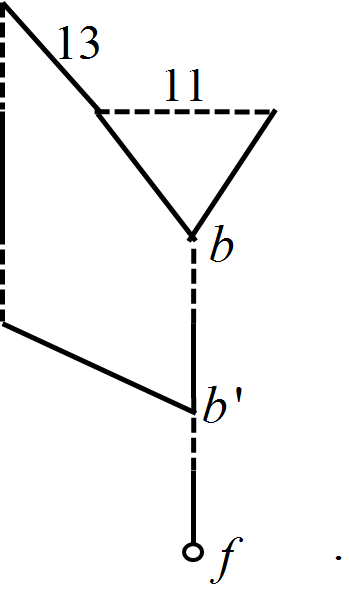}
\caption{$\CB_{b, 11} \subset \CB_{b', 13}$}
\label{fig.lam1}
\end{minipage}
\hspace{4cm}
\begin{minipage}[b]{0.36\linewidth}
\centering
\includegraphics[width=\textwidth]{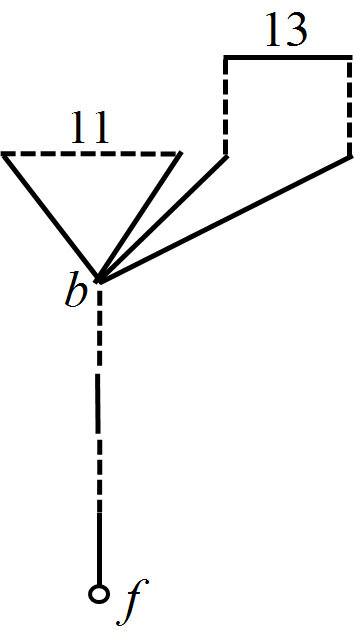}
\caption{$\CB_{b, 11} \subset \CB_{b, 13}$}
\label{fig.lam2}
\end{minipage}
\end{figure}

First let us present an attempt at constructing a counter-example. The subtlety of reason due to which
the counter-example fails should indicate that the proof will be non-trivial. In Figure \ref{fig.lam1}, $\CB_{b, 11} \subset \CB_{b', 13}$
and in Figure \ref{fig.lam2}, $\CB_{b, 11} \subset \CB_{b, 13}$. In Figure \ref{fig.lam3}, we have tried to ``combine'' the blossoms
so that $\CB_{b, 11}$ is contained in both $\CB_{b, 13}$ and $\CB_{b', 13}$ thereby giving a counter-example to laminarity.
However, observe that in the process of ``combining'' the blossoms, the tenacity of $b$ reduces from $\infty$ to 13, so that
$\CB_{b, 13}$ is not a valid blossom anymore.

\definition{\label{def.blossom} (Nesting depth of blossoms)}
Since blossoms were defined recursively, so is their nesting depth.
Let $b$ be an outer vertex and $t$ be an odd number such that $\t(b) > t$ and $t < l_m $. Define the nesting depth of blossom 
$\CB_{b, 1}$ to be $N(\CB_{b, 1}) = 0$.
Define the nesting depth of blossom $\bt$ to be 
\[ N(\bt) = 1 +  \left( \max_{v \in (S_t \cup \{b\}), \ v \ \mbox{outer}} N(\CB_{v, t-2})  \right)  \]
if $S_{b, t} \neq \emptyset$ and $N(\CB_{b, t-2})$ otherwise.

In the graph of Figures \ref{fig.BFSH} and \ref{fig.verten}, the nesting depths of these blossoms $\CB_{a, \alpha}$, $\CB_{b, \kappa}$ 
and $\CB_{f, \tau}$ are 1, 2 and 3, respectively.


\begin{figure}[h]
\begin{center}
\includegraphics[scale = 0.4]{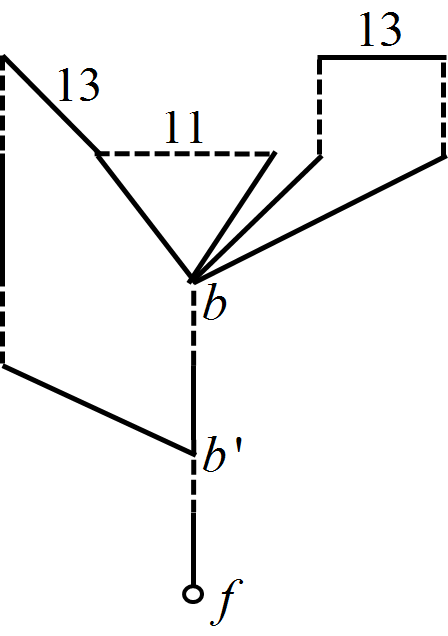}
\caption{Observe that $\t(b) = 13$.}
\label{fig.lam3}
\end{center}
\end{figure}


\begin{lemma}
\label{lem.bb}
Let $v \in \bt$. Then $\exists k$ such that $1 \leq k \leq N(\bt)$ and $b = \base^k(v)$. Furthermore, all the vertices
$\base(v),$ $\ldots , \base^{k-1} (v)$ belong to $\bt$.
\end{lemma}

\begin{proof}
By induction on the nesting depth of blossom $\bt$. If $N(\bt) = 1$, by definition, $b = \base(v)$. 
To prove the induction step, suppose $N(\bt) = l + 1$. Now, if $v \in S_{b, t}$, i.e., $\t(v) = t$, then $\base(v) = b$ and we are done.

Otherwise, $\exists u \in S_{b, t} \cup \{b\}$ such that $v \in \CB_{u, t-2}$. Clearly $N(\CB_{u, t-2}) \leq l$ and either $u = b$ or 
$\base(u) = b$. By the induction hypothesis, $\exists ! k$ such that $l \geq k \geq 1$ and $u = \base^k(v)$. 
If $u = b$, $b = \base^k(v)$, and by the induction hypothesis, $\base(v),$ $\ldots , \base^{k-1} (v)$ belong to $\CB_{b, t-2}$ and
hence also to $\bt$.

If $u \neq b$, $b = \base^{k+1}(v)$ and $k+1 \leq l+1$.
Now, by the induction hypothesis, $\base(v),$ $\ldots , \base^{k-1} (v)$ belong to $\CB_{u, t-2}$.
Hence $\base(v),$ $\ldots , \base^{k} (v)$ belong to $\bt$.
\end{proof}

\begin{lemma}
\label{lem.bb1}
Let $t \leq t' < \t(b)$ and $t' < l_m$, and let
$\bt$ and $\CB_{b, t'}$ be two blossoms with the same base $b$. Then $\bt \subseteq \CB_{b, t'}$.
\end{lemma}

\begin{proof}
The proof is by induction on $t' - t$. The base case, i.e., $t = t'$, is obvious. Assume the induction hypothesis
that $\bt \subseteq \CB_{b, t' - 2}$. Now, by Definition \ref{def.blossom} it is straightforward that $\CB_{b, t' - 2} \subseteq \CB_{b, t'}$.
Hence $\bt \subseteq \CB_{b, t'}$.
\end{proof}

\begin{lemma}
\label{lem.vinb}
Let $\bt$ be a blossom with base $b$ and tenacity $t < l_m$, and $v$ be a vertex satisfying \; $\base^k (v) = b$ for some $k \geq 1$.
If $t \geq \t(base^{k-1} (v))$ then $v \in \bt$.
\end{lemma}

\begin{proof}
The proof is by induction on $k$. In the base case, i.e., $k = 1$, $\base(v) = b$. Let $\t(v) = r$, clearly $r \leq t$.
By Definition \ref{def.blossom}, $v \in \CB_{b, r}$ and by Lemma \ref{lem.bb1}, $\CB_{b, r} \subseteq \bt$.
Hence $v \in \bt$.

For the induction step, let $\base^{k-1} (v) = u$, $\t(u) = r \leq t$. By the induction hypothesis, 
$v \in \CB_{u, r-2}$. Since $\base(u) = b$, by Definition \ref{def.blossom}, $\CB_{u, r-2} \subseteq \CB_{b, r}$ and by Lemma \ref{lem.bb1},
$\CB_{b, r} \subseteq \bt$. Hence $v \in \bt$.
\end{proof}

\begin{lemma}
\label{lem.bb2}
Let $\bt$ and $\CB_{b', t'}$ be two blossoms such that $b \in \CB_{b', t'}$. Then $\bt \subset \CB_{b', t'}$.
\end{lemma}

\begin{proof}
By Lemma \ref{lem.bb}, there is a $k \geq 1$ such that $b' = \base^k(b)$ 
and $\base^1 (b), \ldots, \base^{k-1} (b) \in \CB_{b', t'}$. Clearly, $t' \geq \t(\base^{k-1} (b))$. To prove the statement,
we will apply induction on $k$. 

For the base case, i.e., $k = 1$, let $\t(b) = r$. Clearly, $t < r \leq t'$ and $\bt \subset \CB_{b, r-2}$.
By Definition \ref{def.blossom}, $\CB_{b, r-2} \subset \CB_{b', r}$, where the containment is proper since $b$ is not in the 
first blossom but it is in the second one. By Lemma \ref{lem.bb1}, $\CB_{b', r} \subseteq \CB_{b', t'}$.
Hence $\bt \subset \CB_{b', t'}$.

For the induction step assume $\base^{k+1} (b) = b'$.
Let $\base^{k} (b) = v$, and let $\t(v) = r$. Since $v \in \CB_{b', t'}$, $r \leq t'$. Clearly, $\t(\base^{k-1} (b)) \leq r-2$. 
Therefore, by Lemma \ref{lem.vinb}, $b \in \CB_{v, r-2}$. Furthermore, since $\base^{k} (b) = v$, by the induction hypothesis, 
$\bt \subset \CB_{v, r-2}$.

Since $\base(v) = b'$, by Definition \ref{def.blossom}, $\CB_{v, r-2} \subseteq \CB_{b', r}$. Since $r \leq t'$, 
$\CB_{b', r} \subseteq \CB_{b', t'}$. Hence, $\bt \subset \CB_{b', t'}$.
\end{proof}

\begin{proposition}
\label{prop.laminar}
For eligible tenacity $t$, assume \Ct \ holds. \\
The set of blossoms of tenacity at most $t$ forms a laminar family, i.e., two such blossoms are either disjoint or one is contained in the other.
\end{proposition}

\begin{proof}
Let $t' \leq t$ and $t'' \leq t$.
Suppose $v$ lies in blossoms $\CB_{b, t'}$ and $\CB_{b', t''}$. If $b = b'$, we are done by Lemma \ref{lem.bb1}.
Next assume that $b \neq b'$. Then by the first claim in Lemma \ref{lem.bb}, $b = \base^k(v)$ and $b' = \base^l(v)$, for some
$k$ and $l$. Since $b \neq b'$, $k \neq l$. Let us assume $k < l$.
By the second claim in Lemma \ref{lem.bb}, $b = \base^k(v) \in \CB_{b', t''}$. Finally, by Lemma \ref{lem.bb2}, 
$\CB_{b, t'} \subset \CB_{b', t''}$. Observe that none of the lemmas used assumed existence of blossoms of tenacity $> t$,
hence the proposition follows.
\end{proof}

\section{Proof of Correctness and Running Time}
\label{sec.proof}

We first need to prove that each vertex reachable from an unmatched vertex by an alternating path  will be assigned its correct minlevel and maxlevel. This is done by an induction on the search level in Theorem \ref{thm.levels}. The proof for minlevels is straightforward.

\begin{lemma}
\label{lem.uses}
Let $(u, v)$ be a bridge with $\t(u, v) \leq l_m$.
Then the following hold.
\begin{itemize}
\item
If $(u, v)$ is matched then $u$ and $v$ are both inner vertices.
\item
If $(u, v)$ is unmatched then if $u$ is outer, $\t(u) \leq \t(u, v)$,
and if $u$ is inner, $\t(u) < \t(u, v)$.
\end{itemize}
\end{lemma}

\begin{proof}
First assume that $(u, v)$ is matched. Since $(u, v)$ is not a prop, neither endpoint of this edge assigns a minlevel (of even parity) to the other.
Therefore, the minlevel of both $u$ and $v$ must be odd, and hence they are inner vertices.

Next assume that $(u, v)$ is unmatched. Now, there are three cases:

{\bf Case 1:} $u$ and $v$ are both outer vertices.\\
We will first establish that $\e(u) = \e(v)$. Suppose $\e(u) < \e(v)$.
Then $\e(u) + 1 < \e(v)$. Since an $\e(u)$ path concatenated with edge $(u, v)$ gives an odd alternating path to $v$, we get that 
$\o(v) \leq \e(u) + 1 < \e(v)$, thereby contradicting the assumption that $v$ is outer. 

Hence $\e(u) = \e(v) = i$, say. Clearly, $\o(v) \geq i+1$. Since an $\e(u)$ path concatenated with edge $(u, v)$ gives an odd alternating 
path to $v$ of length $i+1$, $\o(v) = i + 1$. Similarly, $\o(u) = i+1$. Hence $\t(u) = \t(v) = \t(u, v) = 2i+1$.

{\bf Case 2:} $u$ and $v$ are both inner vertices.\\
Since $\mn(u)$ is odd and since $(u, v)$ is not a prop, $\e(v) + 1 > \o(u)$. 
Therefore $\e(u) + \e(v) + 1 > \e(u) + \o(u)$ and hence $\t(u) < \t(u, v)$. Similarly $\t(v) < \t(u, v)$.
   
{\bf Case 3:} $u$ is outer and $v$ is inner.\\   
Since $\mn(v)$ is odd and since $(u, v)$ is not a prop, $\e(u) + 1 > \o(v)$. 
This implies that $\e(u) + \e(v) + 1 > \e(v) + \o(v)$ and hence $\t(v) < \t(u, v)$.
Since an $\e(v)$ path concatenated with edge $(u, v)$ gives an odd alternating path to $u$, we get that 
$\o(v) \leq \e(v) + 1$, and hence $\t(u) \leq \t(u, v)$.
\end{proof}

\begin{remark} In the proof of Lemma \ref{lem.uses}, Case 3, if $\o(v) = \e(v) + 1$ (and hence $\t(u) = \t(u, v)$), the bridge
$(u, v)$ will have non-empty support; in particular, it contains $u$ and its matched neighbor. However, if $\o(v) < \e(v) + 1$, bridge
$(u, v)$ will have empty support.
\end{remark}

\begin{theorem}
\label{thm.levels}
For each vertex $v$ such that $\t(v) < l_m$, Algorithm \ref{alg} assigns $\mn(v)$ and $\mx(v)$ correctly. 
\end{theorem}

\begin{proof}
The case $l_m = 1$ is straightforward and involves finding a maximal matching in $G$. Henceforth we will assume that $l_m \geq 3$.
We will show, by strong induction on $i$, for $i = 0$ to $(l_m -1)/2$ that at search level $i$, Algorithm \ref{alg} will accomplish:
\begin{description}
\item
{\bf Task 1:} Procedure MIN assigns a minlevel of $i+1$ to exactly the set of vertices having this minlevel. It also identifies all props
that assign a minlevel of $i+1$.
\item {\bf Task 2:}
By the end of execution of procedure MIN at this search level, $Br(2i+1)$ is the set of all bridges of tenacity $2i + 1$.
\item {\bf Task 3:}
Procedure MAX assigns correct maxlevels to all vertices having tenacity $2i+1$.
\end{description}

The base case, $i = 0$, is obvious: MIN will assign an oddlevel of 1 to each neighbor of each unmatched vertex. Clearly, no edge
can have tenacity 1.

Next we assume the induction hypothesis for all search levels less than $i$, and prove that Algorithm \ref{alg} will accomplish the three 
tasks at search level $i$.

{\bf Task 1:}
By the induction hypothesis, the minlevel assigned to vertex $v$ at the beginning of execution of MIN at search level $i$ is $\infty$ if and 
only if $\mn(v) \geq i+1$. Since MIN searches from all vertices having level $i$ along the correct parity edges and assigns a minlevel to a vertex 
only if its currently assigned minlevel is $\geq i+1$, any vertex $v$ that is assigned a minlevel in this search level must indeed satisfy 
$\mn(v) = i+1$, and the edge that reaches $v$ will be correctly classified as a prop.

We next prove that every vertex $v$ with $\mn(v) = i+1$ will be assigned its minlevel in this search level, and every prop that assigns a 
minlevel of $i+1$ will be classified as a prop. Let $\mn(v) = i+1$, let $p$ be a $\mn(v)$ path, and let $(u, v)$ be the last edge on $p$. 
Clearly $(u, v)$ is a prop, and every prop that assigns a minlevel of $i+1$ is of this type. Now, $u$ must be BFS-honest w.r.t. $p$: 
If not, then $v$ must occur on a shorter path to $u$, contradicting $\mn(v) < i+1$.
If $|p[f \ \mbox{to} \ u]| = i = \mx(u)$ then $\t(u) < 2i + 1$. Otherwise,
$|p[f \ \mbox{to} \ u]| = i = \mn(u)$. 

In either case, by the induction hypothesis, $u$ has already been assigned a level of $i$.
Therefore, at search level $i$, MIN will search from $u$ along edge $(u, v)$ and will find $v$. 
By the induction hypothesis, at this point, either the minlevel of $v$ is set to either $\infty$ or $i+1$\footnote{The latter case happens if 
$\e(u) = i$ and $v$ has been reached earlier in this search level while searching along an edge $(u', v)$ with $\e(u') = i$.}.
In either case, $v$ will be assigned a minlevel of $i+1$, $u$ will be declared a predecessor of $v$ and $(u, v)$ will be declared a prop.

{\bf Task 2:}
Let $(u, v)$ be a matched bridge with $\t(u, v) = 2i+1$.
By Lemma \ref{lem.t-matched}, $\t(u) = \t(v) = \t(u, v)$, and by Lemma \ref{lem.uses}, 
$u$ and $v$ are both inner. Therefore, $\o(u) = \o(v) = i$.
Hence during search level $i$, MIN will determine that $(u, v)$ is a bridge, that its tenacity is $2i+1$, and will insert it in $Br(2i+1)$.

Next assume that $(u, v)$ is an unmatched bridge with $\t(u, v) = 2i+1$. We will consider the three cases given in Lemma \ref{lem.uses}.
In Case 1, at search level $i$, MIN will determine that $(u, v)$ is a bridge of tenacity $2i+1$. 

In Case 2, assume that $\t(u) \geq \t(v)$; of course $\t(u) < \t(u, v)$. At search level $(\t(v) -1)/2$, MAX will assign $\e(v)$, and
at search level $(\t(u) -1)/2 < i$, while assigning $\e(u)$, MAX will be able to determine that $(u, v)$ is a bridge of tenacity $2i+1$. 

In Case 3, since $u$ is outer, $\t(u) \leq \t(u, v) = 2i+1$. Therefore $\e(u) = \mn(u) \leq i$, and hence it will be
assigned by MIN at search level $\leq i$. MIN will also determine that $(u, v)$ is a bridge.
Since $v$ is inner, $\t(v) < \t(u, v) = 2i+1$. Hence $\e(v) = \mx(v)$ will be assigned by MAX at 
search level $(\t(v) -1)/2 < i$. Of these two operations, the one that happens later will determine the tenacity of bridge $(u, v)$ and
will insert it in $Br(2i+1)$. Clearly, in either case, this will happen by the end of execution of procedure MIN at search level $i$.

{\bf Task 3:}
Statement 2 of Theorem \ref{thm.base}  shows that every vertex of tenacity $2i+1$ lies in the support of a bridge of tenacity $2i+1$, and
by Task 2, all such bridges are in $Br(2i+1)$ at the start of MAX in search level $i$. These two facts together with the following
gives a proof for the current task.

In a run of MAX, consider the point at which DDFS is called with bridge $(u, v) \in Br(2i+1)$. Let $S$ be the set of vertices of tenacity $2i+1$ found 
by MAX so far. We next prove:
\begin{claim}
The set of vertices found by DDFS at this point is $\support(u, v) -S$.
\end{claim}
By the induction hypothesis, every vertex of tenacity $< 2i+1$ is already in a petal. Therefore, as DDFS follows down predecessor edges
starting from $u$ and $v$, if any such vertex is encountered, DDFS will skip to the $\bds$ of this petal \footnote{The importance of 
this subtle point, which is related to the idea of ``precise synchronization of events'' is explained below with the help of Figures \ref{fig.early1} 
and \ref{fig.early2}.
}.
Every vertex in $S$ is also in a petal, hence the same applies to it.

Let $w \in \support(u, v) -S$. Then there is a $\mx(w)$ path containing $(u, v)$, say $p$; assume $p$ starts at unmatched vertex $f$.
Assume $v$ is higher than $u$ on $p$. Then $w$ is reachable from $v$ by following predecessor edges and skipping currently formed petals on the way. 
The path $p[f \ \mbox{to} \ u]$ gives DDFS a disjoint way of reaching ``below'' $w$. Hence DDFS will find $w$.  
\end{proof}

\begin{figure}[ht]
\begin{minipage}[b]{0.5\linewidth}
\centering
\includegraphics[width=\textwidth]{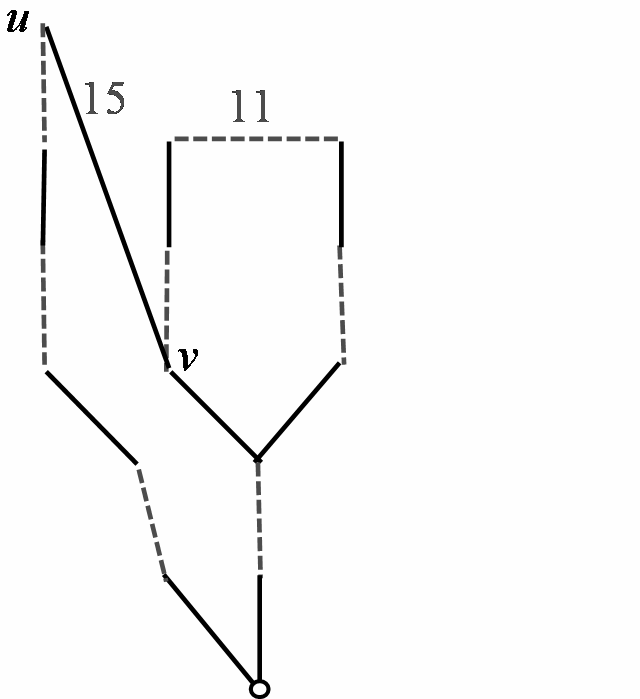}
\caption{Can DDFS be performed on bridge $(u, v)$ at search level 6?}
\label{fig.early1}
\end{minipage}
\hspace{0.5cm}
\begin{minipage}[b]{0.5\linewidth}
\centering
\includegraphics[width=\textwidth]{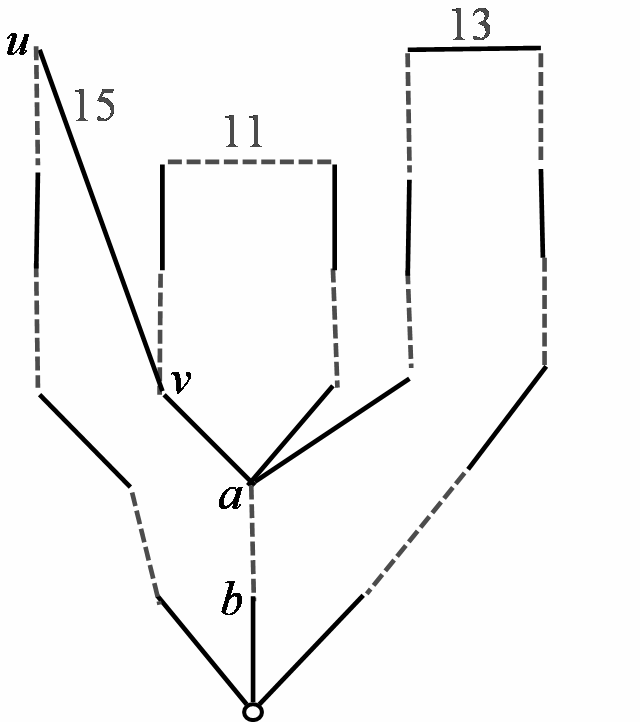}
\caption{If so, vertices $a$ and $b$ will get wrong tenacities.}
\label{fig.early2}
\end{minipage}
\end{figure}

In Figure \ref{fig.early1}, the algorithm determines that $(u, v)$ is a bridge of tenacity 15 at search level 6.
However, according to Algorithm \ref{alg}, DDFS has to be performed on $(u, v)$ at search level 7. The question arises,
``Why wait till search level 7; why not perform DDFS on $(u, v)$ when procedure MAX is run at search level 6?'' To clarify this, let us change 
the algorithm so it runs DDFS on an edge as soon as its tenacity and its status as a bridge have been determined. Assume further
that among the various bridges ready for processing, ties are broken arbitrarily\footnote{By making the example given in Figure 
\ref{fig.early2} slightly bigger, one can easily ensure that there are no such ties.}.

Now consider the enhanced graph of Figure \ref{fig.early2}, in which vertices $a$ and $b$ are in the support of the bridge of 
tenacity 13 and hence have tenacity 13. Assume that when MAX is run at search level 6 the bridge of tenacity 15 is processed first. 
Since the tenacities of vertices $a$ and $b$ are not set yet, DDFS will visit them and assign them a of tenacity 15, which would be incorrect. 
Observe that the correctness of MAX crucially depends on assigning tenacities to each edge of tenacity less than $2i+1$ before
processing bridges of tenacity $2i+1$, i.e., the precise manner in which events are synchronized in Algorithm \ref{alg}. 

\begin{lemma}
\label{lem.maximal}
The procedures given in Section \ref{sec.finding} will find a maximal set of disjoint minimum length augmenting paths in $G$.
\end{lemma}

\begin{proof}
By the properties established in Theorem \ref{thm.base}, the procedure given in Section \ref{sec.one} finds a minimum length augmenting path, say $p$. Clearly, the vertices removed by procedure RECURSIVE REMOVE of Section \ref{sec.maximal} cannot be part of a minimum length augmenting path that is disjoint from $p$. Suppose $v \in p$ is in a blossom and let $\CB$ be the maximal such blossom, with base $b$. Then, by Theorem \ref{thm.base}, $b \in p$ and it will be removed from the graph. Following this, RECURSIVE REMOVE will remove all remaining vertices of $\CB$. Therefore, w.r.t. any remaining bridge, the remaining graph will satisfy the DDFS Requirement. Hence the next path, if any, can be found by the same process. 
\end{proof}

\begin{theorem}
\label{thm.time}
The MV algorithm finds a maximum matching in general graphs in time $O(m \sqrt{n})$ on the RAM model and 
$O(m \sqrt{n} \cdot \alpha(m, n))$ on the pointer model, where $\alpha$ is the inverse Ackerman function.
\end{theorem}

\begin{proof}
Each of the procedures of MIN, MAX, finding augmenting paths, and RECURSIVE REMOVE examine each edge a constant number of times in each phase. The only operation that remains is
that of computing $\bds$ during DDFS. This can be implemented on the pointer model using the set union algorithm \cite{Tarjan}, which will take $O(m \cdot \alpha(m, n))$ time per phase. Alternatively, it can be implemented on the RAM model using the linear time algorithm for a special case of set union \cite{GTarjan}; this will take $O(m)$ time per phase. Since $O(\sqrt{n})$ phases suffice for finding a maximum matching \cite{Karp,Karzanov}, the theorem follows.
\end{proof}

A question arising from Theorem \ref{thm.time} is whether there is a linear time implementation of $\bds$ in the pointer model. \cite{MV} had claimed, without proof, that path compression by itself suffices to achieve this. They stated that because of the special structure of blossoms, a charging argument could be given that assigns a constant cost to each edge. This claim is left as an open problem. 

Next, let us consider the question, ``What is the best way of implementing $\bds$ computations in practice?'' To answer this, let us compare an 
implementation based on the set union data structure \cite{Tarjan} and an implementation that uses only path compression. In the latter case, there is 
no need to build and maintain a separate data structure on the side: each bud simply
maintains and updates a pointer to the lowest bud it has reached. In the former case, not only is a separate data structure needed, but the
``trees'' obtained in it will, in general, destroy the natural tree structure of nesting of petals. Additionally, in the latter case, one would expect the unions 
to be close to balanced anyway in practice. Considering the implementation effort and computational overhead of the former approach,
we believe the latter one is superior. Using \cite{TarjanV}, the worst case running time for path compression in a phase in the second case is bounded by $O(m \log(n))$.



	\section{Discussion}
\label{sec.discussion}

For the bipartite graph maximum matching problem, a few years ago, the running time was improved to $O(m^{10/7})$ \cite{Madry13} and more recently to $O(m^{11/8})$ \cite{Sidford-L}; both these improvements are for sparse graphs only. A concerted effort has been made to improve the running time for the case of general graphs as well. 

As is well known, general graph matching has numerous applications. We mention a particularly interesting and important application to kidney exchange. This was mentioned in the Scientific Background for the 2012 Nobel Prize in Economics awarded to Alvin Roth and Lloyd Shapley \cite{Nobel}. Assume that agent $A$ requires a kidney transplant and agent $B$ has agreed to donate one of her kidneys to $A$; however, their kidney type is not compatible. Assume further that $(A', B')$ is another pair with an incompatibility. If it turns out that $(A, B')$ and $(A', B)$ are both compatible pairs, then let us say that the two pairs are {\em consistent}. If two pairs are consistent, then the appropriate transplants can be performed. 

Next assume that a number of incompatible pairs are specified, $(A_1, B_1), \ldots, (A_n, B_n)$ and for every two pairs, we know whether they are consistent. The problem is to find the maximum number of disjoint consistent pairs. This can be reduced to maximum matching as follows. Let $G = (V, E)$ be a graph with $V = \{v_1, \ldots , v_n\}$ where $v_i$ represents the pair $(A_i, B_i)$, and $(v_i, v_j) \in E$ if and only if the pairs $(A_i, B_i)$ and $(A_j, B_j)$ are consistent. Clearly a maximum matching in $G$ will yield the answer. 
		
	\section{Acknowledgements}
\label{sec.ack}

I wish to thank Ruta Mehta for diligently helping verify this proof, and Silvio Micali for embarking on a year-long journey, full of youthful exuberance, which led to the discovery of this algorithm. For the four-decade-long journey that led to the discovery of this proof, I must thank matching theory for its gloriously elegant combinatorial structure, which kept me going.

	\bibliographystyle{alpha}
	\bibliography{refs}
	
	
\end{document}